\documentclass[11pt]{article} 

\usepackage[margin=1in]{geometry} 
\usepackage{libertine}

\usepackage{comment,xspace} 

\usepackage[dvipsnames]{xcolor} 

\usepackage[ruled,linesnumbered,vlined]{algorithm2e}
\usepackage{algorithmic}
\usepackage{comment,xspace,enumitem}
\usepackage{booktabs}
\usepackage{wrapfig}

\usepackage{changepage}
\usepackage{booktabs,multirow,subcaption}
\usepackage[nice]{nicefrac}

\usepackage[breakable, theorems, skins]{tcolorbox}

\usepackage{amsmath,amsfonts,amssymb,amsthm,mathtools,thmtools} 
\DeclareMathOperator*{\argmax}{arg\,max}
\DeclareMathOperator*{\argmin}{arg\,min}

\usepackage{dsfont}

\usepackage[square,numbers,sort]{natbib} 

\usepackage{hyperref}
\hypersetup{
linktocpage,
colorlinks=true,
citecolor=cobalt, 
urlcolor=ptblue, 
linkcolor=cobalt, 
}
\usepackage{cleveref} 

\usepackage{bbold}

\usepackage{tikz} 
\theoremstyle{plain}
\newtheorem{theorem}{Theorem}[section]
\newtheorem{corollary}[theorem]{Corollary}

\newtheorem{lemma}[theorem]{Lemma}

\newtheorem{claim}[theorem]{Claim}

\theoremstyle{definition}
\newtheorem{definition}[theorem]{Definition}

\declaretheorem[style=definition,qed=$\bigtriangleup$,sibling=theorem]{example}

\newtheorem*{theorem*}{Theorem}

\theoremstyle{remark}
\newtheorem*{remark}{\upshape\bfseries Remark}

\definecolor{cobalt}{rgb}{0.0, 0.28, 0.67}

\newcommand{\expval}[1]{\mathbb{E}\left[#1\right]}
\newcommand{\prob}[1]{\Pr \left\{#1\right\}}
\newcommand{\ellone}[1]{\left\|#1\right\|_1}
\newcommand{\ellinfty}[1]{\left\|#1\right\|_\infty}
\newcommand{\ellzed}[1]{\left\|#1\right\|_z}
\newcommand{\mmshat}{\widehat{\mathrm{MMS}}}
\newcommand{\muhat}{\widehat{\mu}}

\newcommand{\calA}{\mathcal{A}}
\newcommand{\calB}{\mathcal{B}}
\newcommand{\calM}{\mathcal{M}}

\newcommand{\calI}{\mathcal{I}}
\newcommand{\calP}{\mathcal{P}}

\newcommand{\calQ}{\mathcal{Q}}

\newcommand{\calR}{\mathcal{R}}
\newcommand{\calF}{\mathcal{F}}

\newcommand{\rmchar}{\mathrm{char}}

\title{\bfseries Exact Maximin Share Fairness via Adjusted Supply}

\author{Siddharth Barman\thanks{Indian Institute of Science, Bangalore;  {barman@iisc.ac.in} }  \and Satyanand Rammohan\thanks{Indian Institute of Science, Bangalore;  {satyanand.rammohan@gmail.com} }  \and Aditi Sethia\thanks{Indian Institute of Science, Bangalore;  {aditisethia@iisc.ac.in} }}

\date{}

\begin{document}
\maketitle

\begin{abstract}
This work addresses fair allocation of indivisible items in settings wherein it is feasible to create copies of resources (goods) or dispose of tasks (chores). 
We establish that exact maximin share (MMS) fairness can be achieved via limited duplication of goods even under monotone valuations. We also show that, when allocating indivisible chores under monotone costs, MMS fairness is always feasible with limited disposal of chores. Since monotone valuations (and monotone costs) do not admit any nontrivial approximation guarantees for MMS, our results highlight that such barriers for MMS can be circumvented by post facto adjustments in the supply of the given items. 

We prove that, for division of $m$ goods among $n$ agents with monotone valuations, there always exists an assignment of subsets of goods to the agents such that they receive at least their maximin shares and no single good is allocated to more than $3 \log m$ agents. In addition, the sum of the sizes of the assigned subsets---i.e., the total number of goods assigned, with copies---does not exceed $m$. We also address instances with identically ordered valuations. In such instances, the marginal values of the goods conform to a common order across all the agents and the subsets. For identically ordered valuations, we obtain an upper  bound of $O(\sqrt{\log m})$ on the maximum assignment multiplicity across goods and an $m +  \widetilde{O}\left(\frac{m}{\sqrt{n}} \right)$ bound for the total number of goods assigned. Further, for additive valuations, we prove that there always exists an MMS assignment in which no single good is allocated to more than $2$ agents and the total number of goods assigned, with copies, is at most $2m$.

For fair division of chores, we upper bound the number of chores that need to be discarded for ensuring MMS fairness. We prove that, under monotone costs, there always exists an MMS assignment in which, out of the $m$ given chores, at most $\frac{m}{e}$ remain unassigned. 
For identically ordered costs, we establish that MMS fairness can be achieved while keeping at most $\widetilde{O} \left(\frac{m}{n^{1/4}} \right)$ chores unassigned. We further show that, under additive costs, one can ensure MMS fairness by discarding at most $\frac{2m}{11} + n$ chores. We also prove that the obtained bounds for monotone valuations and monotone costs are essentially tight. 
\end{abstract}
\section{Introduction}
Maximin shares constitute a fundamental criterion when studying fair allocation of indivisible goods among agents with individual preferences \cite{hill1987partitioning,7c65302b-f079-361a-94f1-0c3c9f6fc76b}. Among $n$ agents, these shares are defined by considering a cut-and-choose thought experiment: Each agent partitions the grand bundle of goods into $n$ subsets and receives the least valued one, according to its valuation. The maximin share (MMS) is the maximum value the agent can guarantee for itself in this exercise. Further, an allocation of the indivisible goods among the agents is deemed to be (maximin share) fair if each agent receives a bundle of value at least as much as its maximin share. 

While conceptually central in the context of fair division of indivisible items \cite{amanatidis2023fair}, MMS fairness is not a feasible criterion: there exist fair division instances---even with agents that have additive valuations over the goods---wherein, via any allocation, one cannot ensure maximin shares simultaneously for all the agents. That is, there exist instances with additive valuations that do not admit MMS fair allocations; see \cite{10.1145/3140756,10.1145/2600057.2602835,feige2021tight}, and references therein. In light of the infeasibility of exact MMS fairness, multiple prior works have instead focused on approximation guarantees. These guarantees ensure that each agent receives a bundle of value at least $\alpha \in (0,1)$ times its maximin share for as large a value of $\alpha$ as possible. The best-known approximation bound gets progressively worse as one moves up valuation classes. In particular, this is $\alpha = \frac{3}{4} + \frac{1}{12n}$ for additive valuations \cite{10.1145/3391403.3399526}. For submodular and XOS valuations the best-known $\alpha$ is $\frac{1}{3}$ \cite{10.1145/3219166.3219238, barman2020approximation} and $\frac{1}{4.6}$ \cite{SEDDIGHIN2024104049}, respectively. Further, the approximation bounds for subadditive valuations are logarithmic in the number of agents \cite{SEDDIGHIN2024104049} or the number of goods \cite{10.1145/3219166.3219238}. In fact, monotone valuations do not admit any nontrivial approximation guarantee; specifically, there exist instances with two agents that have monotone valuations over four goods such that the maximin shares of both the agents is one, and under any allocation, one of the agents necessarily receives a bundle of value zero. 

We bypass these impossibility results and achieve exact MMS fairness by post facto adjustments in the supply of items. In particular, this work establishes that one can achieve exact MMS fairness via limited duplication of goods, even under monotone valuations. We also address fair chore division and show that, with limited disposal of chores, MMS fairness is feasible under monotone costs. Note that the obtained MMS fairness guarantees are with respect to the shares defined in the underlying instance. That is, we start with MMS shares, as defined in the given instance, and quantify the extent to which goods have to be duplicated, or chores have to be disposed of, to achieve fairness. 

The current framework is applicable in allocation settings wherein it is feasible to create copies of resources (goods) or dispose of tasks (chores). For instance, the work of Budish \cite{7c65302b-f079-361a-94f1-0c3c9f6fc76b} conforms to a similar treatment in the context of fair course allocation; see also \cite{budish2017course} and the {\it ScheduleScout} platform.\footnote{{https://www.getschedulescout.com/}} Indeed, class strengths of popular courses can often be increased to an extent. Duplication is also possible when allocating scalable resources, such as medical supplies and food donations. Another relevant example here is the fair assignment of students to schools. In this context, prior works \cite{procaccia2024school, santhini2024approximation} have shown that fairness---in particular, proportionality among demographic groups---can be achieved by increasing the capacities of selected schools. Capacity modification has been studied for other social choice objectives as well, including stability (see, e.g., \cite{nguyen2018near,gokhale2024capacity}) and Pareto optimality \cite{kumano2022quota}.

\paragraph{Our Results.} 
We now summarize our contributions for fair division of 

\noindent
(i) Goods: These are indivisible items desired by the agents. Hence, in this context, we consider the agents' valuations to be (set-inclusion-wise) monotone non-decreasing.  

\noindent
(ii) Chores: Here the items to be assigned are undesirable, i.e., have monotone costs for the agents. 

In the standard discrete fair division setup, the subsets of indivisible items assigned to the agents are pairwise disjoint. By contrast, in the current framework, a limited number of goods can be duplicated, or some chores can be discarded. Hence, the assigned bundles are not necessarily disjoint or required to partition the underlying set of items. We use the term \emph{multi-allocation} $\calA = (A_1, \ldots, A_n)$ to highlight this aspect; here, $A_i$ is the subset of items assigned to agent $i$ and, as mentioned, some items may be present in multiple subsets or even in none. It is relevant to note that, throughout, we never assign multiple copies of the same item to the same agent. Hence, in any multi-allocation, each bundle $A_i$ is a subset and not a multiset of items. This model feature ensures that the agents' valuations defined over subsets of items in the underlying instance continue to be well-defined even when considering multi-allocations. 

The overarching goal of this work is to establish the existence of multi-allocations that provide each agent its maximin share and, at the same time, have bounded duplication of goods or disposal of chores. Towards a quantitative instantiation of this goal, we will consider the characteristic vectors $\chi^\calA \in \mathbb{Z}^m_+$ of multi-allocations $\calA =(A_1, \ldots, A_n)$;\footnote{This is a slight abuse of terminology, since the components of $\chi^\calA$ are integers and not necessarily zero or one.} here, $m$ denotes the total number of items in the given instance and the $j$th component of this vector denotes the number of copies of item $j \in [m]$ assigned among the bundles $A_1, \ldots, A_n$.    

For the case of goods, we establish the existence of MMS multi-allocations $\calA=(A_1, \ldots, A_n)$ with bounds on the $\ell_1$ and $\ell_\infty$ norms of the characteristic vectors $\chi^\calA$. Note that the $\ell_1$ norm, $\ellone{\chi^\calA}$, captures the total number of goods, with copies, assigned among the agents,  $\ellone{\chi^\calA} = \sum_{i=1}^n |A_i|$. In addition, $\ellinfty{\chi^\calA} = \max_{j \in [m]} \ \chi^\calA_j$ captures the maximum number of copies of any one good assigned under $\calA$. Hence, the $\ell_1$ and $\ell_\infty$ bounds quantify the extent of duplication of goods required to achieve MMS fairness. 

For chores, we upper bound the number of zero components of $\chi^\calA$. Specifically, we write $\ellzed{\chi^\calA} \coloneqq \{ j \in [m] \mid \chi^\calA_j = 0 \}$ and note $\ellzed{\chi^\calA}$ corresponds to the number of unassigned chores in any multi-allocation $\calA$. 

\begin{table}[!h]
\centering
\begin{tabular}{|c|cc|}
\hline
         & \multicolumn{2}{|c|}{\textbf{Goods}}                                                                                                                                                                                                                                \\ \hline         & \multicolumn{1}{c|}{\begin{tabular}[c]{@{}c@{}}Maximum multiplicity  \\  $\|\chi^\calA\|_\infty$\end{tabular}} & \begin{tabular}[c]{@{}c@{}}Total number of assigned goods (with copies) \\ $\|\chi^\calA\|_1$ \end{tabular} \\ \hline
         
{\begin{tabular}[c]{@{}c@{}} Monotone Valuations \\ (Theorem \ref{thm:monotone_goods}) \end{tabular}} & \multicolumn{1}{c|}  {\begin{tabular}[c]{@{}c@{}}{${3 \log m}^\star$}  \end{tabular}}
& {\begin{tabular}[c]{@{}c@{}}$m$ \end{tabular}}                                                                                                                                              \\ \hline
{\begin{tabular}[c]{@{}c@{}} Identically Ordered Valuations \\ (\Cref{thm:additive_goods}) \end{tabular}} & \multicolumn{1}{c|}{{\begin{tabular}[c]{@{}c@{}}{$12 \sqrt{ \log m}$} \end{tabular}}}                                                                       & \multicolumn{1}{c|}{{\begin{tabular}[c]{@{}c@{}}{$m + \frac{6 m \sqrt{\log m}}{\sqrt{n}}$} \\\end{tabular}}} 
                                                                  \\ \hline

{\begin{tabular}[c]{@{}c@{}}  Additive Valuations \\ (\Cref{thm:additive_3_3m}) \end{tabular}}  & \multicolumn{1}{c|}{\begin{tabular}[c]{@{}c@{}}{$2$} \end{tabular}}                                                                                    &  {\begin{tabular}[c]{@{}c@{}}{$2m$} \end{tabular}}                                                                                                                                                                    \\ \hline

{\begin{tabular}[c]{@{}c@{}} Monotone Valuations \\ with Entitlements \\ (Theorem \ref{thm:mms_hat_monotone_goods}) \end{tabular}} & \multicolumn{1}{c|}  {\begin{tabular}[c]{@{}c@{}}{${3 \log m}$}  \end{tabular}}
& {\begin{tabular}[c]{@{}c@{}}$\lceil 1.7 m \rceil$ \end{tabular}}                                                                                                                                              \\ \hline
\end{tabular}
\caption{The table presents upper bounds on the duplication of goods required to achieve MMS fairness. 
The result marked $^\star$ is essentially tight; \Cref{thm:goods_lowerbound} provides a matching lower bound.}
\label{tab:results}
\end{table}

\paragraph{Goods.} Our results for fair division of goods are listed below; see also Table \ref{tab:results}. 
\begin{itemize}
\item For monotone valuations, we prove that exact MMS fairness can be achieved while ensuring that, for each good, the number of copies assigned is at most $3 \log m$ and the total number of allocated goods, including copies, does not exceed the underlying count $m$. That is, there always exists an MMS multi-allocation $\calA$ whose characteristic vector has $\ell_{\infty}$ norm at most $3 \log m$ and $\ell_1$ norm at most $m$. The relevant implication of this result is that, even for the general setting of monotone valuations, one can achieve exact MMS fairness while ensuring that the total number of allocated goods does not exceed the given count $m$, and at most logarithmically many copies of any single good are required. We also show that here the $\ell_{\infty}$ bound is essentially tight by establishing a matching lower bound. 

This result for monotone valuations is obtained via the probabilistic method: We start with MMS-inducing partitions $\calM^i = (M^i_1, \ldots, M^i_n)$ for each agent $i$. Then, we sample a bundle $R_i$ uniformly at random from among the subsets $M^i_1, \ldots, M^i_n$ and show that the random multi-allocation $ (R_1, \ldots, R_n)$ satisfies the stated bounds with positive probability. The proof highlights that one can efficiently find MMS multi-allocations, with the stated $\ell_\infty$ and $\ell_1$ bounds, if the MMS-inducing partitions, $\calM^i$, for all the agents $i$ are given as input. At the same time, minimizing the $\ell_\infty$ norm in this setup is {\rm NP}-hard (\Cref{sec:reduction}).  

\item This work also obtains specialized bounds for instances in which the marginal values of the goods conform to a common order across all the agents and the subsets. In particular, we say that the valuations are identically ordered if there exists an indexing of the $m$ goods, $\{g_1, \ldots g_m\}$, such that for each pair of goods $g_s, g_t \in [m]$, with index $s < t$, and any $S \subset [m]$ that does not contain $g_s$ and $g_t$, each  agent values $S \cup \{g_s \}$ at least as much as $S \cup \{g_t \}$. For such identically ordered valuations, we establish an $\ell_{\infty}$ bound of $O(\sqrt{\log m})$ and an $\ell_1$ bound of $m +  \widetilde{O}\left(\frac{m}{\sqrt{n}} \right)$; where, $\widetilde{O}(\cdot)$ subsumes logarithmic dependencies. Note that here the $\ell_1$ norm additively exceeds the given number of goods, $m$, by only a lower-order term.

\item For additive valuations,\footnote{Recall that a set function $v$ is said to be additive if the value, $v(S)$, of any subset $S$, is equal to the sum of the values of the elements in $S$.} we establish the existence of MMS multi-allocations $\calA$ in which the $\ell_{\infty}$ norm is upper bounded by $2$ and the $\ell_1$ norm is at most $2m$. 

\item We also address MMS fairness among agents with arbitrary entitlements; see \Cref{section:entitlements} for details. For this setting and under monotone valuations, we obtain upper bounds of $3 \log m$ and $\lceil 1.7 m \rceil$ for $\ell_\infty$ and $\ell_1$, respectively. 
 \end{itemize}

\paragraph{Chores.} As mentioned previously, in the context of chores, we upper bound the number of chores that need to be discarded for ensuring MMS fairness, i.e., we quantify the number of zero components, $\ellzed{\chi^\calA}$, for MMS multi-allocations $\calA$. Note that the notion of keeping some chores unassigned is analogous to the construct of charity, which has been studied in the context of the envy-freeness up to any good (EFx); see, e.g., \cite{chaudhury2021little}. Though, in contrast to prior works on fair division of goods with charity, the current paper focuses on chores and maximin shares.  We next list our results for fair division of chores; see also Table \ref{tab:chores}.
\begin{table}
\centering
\begin{tabular}{|c|c|}
\hline
                 & \textbf{Chores}                                                                                                                 \\ \hline
                 & \begin{tabular}[c]{@{}c@{}}Total number of unassigned chores \\ $\|\chi^\calA\|_z$\end{tabular}                                    \\ \hline
{\begin{tabular}[c]{@{}c@{}} Monotone Costs \\ (\Cref{thm:monotone_chores}) \end{tabular}}        & \begin{tabular}[c]{@{}c@{}}${\frac{m}{e}}^\star$ \end{tabular}                  \\ \hline
{\begin{tabular}[c]{@{}c@{}} Identically Ordered Costs \\ (\Cref{thm:additive_chores}) \end{tabular}} & \begin{tabular}[c]{@{}c@{}} $\frac{3m \left( \log m \right)^{\nicefrac{3}{2}}}{n^{\nicefrac{1}{4}}}$ \end{tabular} \\ \hline
{\begin{tabular}[c]{@{}c@{}} Additive Costs \\ (\Cref{thm:chores-additive}) \end{tabular}} & \begin{tabular}[c]{@{}c@{}} $\frac{2m}{11} + n$ \end{tabular} \\ \hline
\end{tabular}
\caption{The table presents upper bounds on the total number of unassigned chores for MMS fairness. The result marked $^\star$ is essentially tight; \Cref{thm:chores_lowerbound} provides a matching lower bound.} 
\label{tab:chores}
\end{table}

 \begin{itemize}
 \item Under monotone costs, we show that there always exists an MMS multi-allocation in which, out of the $m$ given chores, at most $\frac{m}{e}$ remain unassigned. We complement this upper bound by showing that there exist instances, with monotone costs, such that in every MMS multi-allocation, at least $(1- o(1))\frac{m}{e}$ chores remain unassigned.
\item For identically ordered costs, 
we establish that MMS fairness can be achieved while keeping $\widetilde{O} \left(\frac{m}{n^{1/4}} \right)$ chores unassigned. 
\item We further show that, under additive costs, one can ensure MMS fairness by discarding at most $\frac{2m}{11} + n$ chores
\end{itemize} 
 
\paragraph{Additional Related Work.}
As mentioned, Budish \cite{7c65302b-f079-361a-94f1-0c3c9f6fc76b} also considered the duplication of goods to achieve fairness. However, the result in \cite{7c65302b-f079-361a-94f1-0c3c9f6fc76b} confines to $1$-out-of-$(n+1)$ MMS, which is an ordinal approximation, and holds for additive (cancellable) valuations. The current work extends to exact MMS and covers monotone valuations. Further, the techniques utilized in this work are significantly different from the approach in \cite{7c65302b-f079-361a-94f1-0c3c9f6fc76b} of approximate competitive equilibrium from equal incomes. In fact, to the best of our knowledge, this work represents the first extensive use of probabilistic methods in the context of fair item allocation.

For fair division of chores, it was shown in \cite{Aziz_Rauchecker_Schryen_Walsh_2017} that there exist instances with additive costs that do not admit MMS fair allocations. Towards approximation guarantees, \cite{Aziz_Rauchecker_Schryen_Walsh_2017} showed that a $2$-approximately MMS chore division always exists under additive costs. This approximation bound was improved to $\frac{4}{3}$ \cite{barman2020approximation} and, subsequently, to $\frac{11}{9}$ \cite{10.1145/3465456.3467555}.
\section{Notation and Preliminaries}\label{sec:prelims}
This section fixes the notation and relevant notions for fair division of goods; the notation specific to division of chores is relegated to Section \ref{sec:chores}. 
 
\paragraph{Fair Division Instances.} A {fair division instance} is given by a tuple $\langle [n], [m], \{v_i\}_{i=1}^n \rangle$, where $[n]=\{1,2,.\dots,n\}$ is the set of $n\in\mathbb{Z}_+$ agents, $[m]=\{1,2, \dots, m\}$ the set of $m\in \mathbb{Z}_+$ indivisible goods, and for each agent $i\in[n]$, the set function $v_i: 2^{[m]} \to \mathbb{R}_+$ denotes the valuation of agent $i$ over subsets of goods. Specifically, $v_i(S) \in \mathbb{R}_+$ denotes the value that agent $i$ derives from the subset $S \subseteq [m]$ of goods. For subsets $S \subseteq [m]$ and $g \in [m]$, we will write $S + g$ to denote the union $S \cup \{ g\}$. 

A valuation $v_i$ is said to be monotone if the inclusion of goods into any subset does not decrease its value, under $v_i$, i.e., $v_i(S)\leq v_i(T)$ for every pair of subsets $S \subseteq T \subseteq[m]$. We will assume throughout that the agents' valuations are monotone and normalized: $v_i(\emptyset)=0$ for all agents $i$. 

We will additionally consider instances with identically ordered valuations. Here, we have an indexing of the $m$ goods, $\{g_1, \ldots g_m\}$, such that for each pair of goods $g_s, g_t$, with index $s < t$, and all agents $i \in [n]$, the inequality $v_i(S + g_s) \geq  v_i(S + g_t)$ holds for each subset $S \subset [m]$ that does not contain $g_s$ and $g_t$; see Example \ref{ex:sqrt-ordered} in Section \ref{subsec:additive-ordered}. 

This work also establishes improved bounds for the specific case of additive valuations. Recall that a valuation $v_i$ is said to be additive if, for every subset $S\subseteq[m]$ of goods, $v_i(S)=\sum_{g\in S} v_i(\{g\})$. We will use the shorthand $v_i(g)$---instead of $v_i(\{g\}) \in \mathbb{R}_+$---to denote agent $i$'s value for any good $g \in [m]$.

\paragraph{Allocations and Multi-Allocations.} An allocation $\calB=(B_1,B_2,\ldots, B_n)$ of the goods among the $n$ agents is a partition of $[m]$ into $n$ pairwise disjoint subsets $B_1,\ldots, B_n \subseteq [m]$. Here, the subset of goods $B_i$ is assigned to agent $i \in [n]$ and is referred to as $i$'s bundle. In addition, write $\Pi_n([m])$ to denote the collection of all $n$-partitions of $[m]$. Note that for any allocation $\calB =(B_1,\ldots, B_n)$ we have, by definition, $\cup_{i=1}^n B_i = [m]$ and $B_i \cap B_j = \emptyset$, for all $i \neq j$, and hence $\calB \in \Pi_n([m])$.

A \textit{multi-allocation} is a tuple $\calA=(A_1,A_2\dots,A_n)$ of $n$ subsets, wherein subset $A_i \subseteq [m]$ denotes the bundle assigned to agent $i$. In contrast to allocations, in a multi-allocation, we do not require that the assigned bundles $A_i$ are pairwise disjoint and that they partition $[m]$.\footnote{Note that $A_i$s are still subsets of goods and not multisets.} Hence, in a multi-allocation, a single good may be present in multiple bundles or even in none. 

Though, when in a multi-allocation $\calA$, each good $g$ is assigned to exactly one agent, we refer to $\calA$ as an {\it exact allocation}; this is to emphasize that the bundles of such a multi-allocation do partition $[m]$. 

We associate with each bundle $A_i \subseteq [m]$ an $m$-dimensional characteristic vector $\rmchar(A_i) \in \{0,1\}^m$. For each good $g\in [m]$, the $g$th component of the characteristic vector---denoted as $\rmchar(A_i)_g$---is equal to one if $g \in A_i$, otherwise the $g$th component is zero. That is, 
\begin{align*}
\rmchar(A_i)_g \coloneqq \begin{cases}
    1 & \text{if } g\in A_i \\
    0 & \text{otherwise}.
\end{cases}
\end{align*}

For any multi-allocation $\calA=(A_1, \ldots, A_n)$, we will use $\chi^\calA \in \mathbb{Z}^m_+$ to denote the vector sum of the characteristic vectors of its bundles, $\chi^\calA \coloneqq \sum_{i=1}^n\rmchar(A_i)$. We will refer to $\chi^\calA$ as the \textit{characteristic vector} of the multi-allocation $\calA$. When there is no ambiguity, we will omit the notational dependence in the superscript and simply write $\chi$ for $\chi^\calA$.

Note that for any good $g\in [m]$ and multi-allocation $\calA$, the $g^{th}$ component of the characteristic vector $\chi^\calA_g$ is equal to the number of bundles in $\calA$ that contain $g$. Conceptually, we think of this setting as one in which $\chi^A_g$ identical copies of the good $g$ are assigned among different agents. 

Write $\ellone{\chi^\calA}$ and $\ellinfty{\chi^\calA}$ to denote the $\ell_1$ and $\ell_\infty$ norm, respectively, of the characteristic vector. Hence,  $\ellone{\chi^\calA} = \sum_{g=1}^m \chi^\calA_g$ and $\ellinfty{\chi^\calA} = \max_{g\in[m]} \chi^\calA_g$. It is relevant to note that $\ellone{\chi^\calA}$ captures the total number of goods, with copies, assigned among the agents,  $\ellone{\chi^\calA} = \sum_{i=1}^n |A_i|$. Further, $\ellinfty{\chi^\calA}$ captures the maximum number of copies of any one good $g$ assigned under $\calA$.

In particular, if $\calA$ is an {\it exact} allocation, then $\chi^\calA$ is equal to the all-ones vector and we have $\ellone{\chi^\calA} =m$ and $\ellinfty{\chi^\calA} =1$.
 
\noindent
The shared-based fairness criterion we study in this work is defined using maximin shares; these shares are defined next.
\begin{definition}[Maximin Share (MMS)]\label{def:mms}
    Given any fair division instance $\langle [n], [m], \{v_i\}_{i=1}^n \rangle$ with goods, the {maximin share}, $\mu_i \in \mathbb{R}_+$, of each agent $i \in [n]$ is defined as 
    \begin{align*}
    \mu_i \coloneqq  \max_{(X_1,\dots, X_n) \in \Pi_n([m])} \ \ \min_{j\in[n]} v_i(X_{j}).
    \end{align*}
Further, for each agent $i$, let $\calM^i=(M^i_1, M^i_2, \ldots, M^i_n) \in \Pi_n([m])$ denote an {MMS-inducing partition}:
\begin{align*}
\calM^i \in \argmax_{(X_1,\dots, X_n) \in \Pi_n([m])} \ \ \min_{j\in[n]} v_i(X_{j})
\end{align*}
\end{definition}

Note that in Definition \ref{def:mms} the maximum is taken over all $n$-partitions of $[m]$. Also, by definition, the partition $\calM^i =(M^i_1, \ldots, M^i_n)$ satisfies $v_i(M^i_j) \geq \mu_i$, for each index $j \in [n]$. 

\paragraph{Fair Multi-Allocations.} A multi-allocation $\calA=(A_1,\dots,A_n)$ is said to be an \emph{MMS multi-allocation} (i.e., it is deemed to be fair) if under it each agent receives a bundle of value at least its maximin share:  $v_i(A_i)\geq \mu_i$ for all agents $i \in [n]$.
 
To establish existential guarantees for MMS multi-allocations $\calA$, we will assume that, for all the agents, we are given the MMS-inducing partitions $\calM^i$, which in turn are guaranteed to exist (see Definition \ref{def:mms}).  
\section{Fair Division of Goods}\label{sec:goods}

\subsection{Monotone Valuations}
This section addresses fair division of goods under monotone valuations. Our positive result for this setting is stated below.

\begin{theorem}
\label{thm:monotone_goods}
Every fair division instance $\langle [n], [m], \{v_i\}_{i=1}^n \rangle$ with monotone valuations admits an MMS multi-allocation $\calA$ in which no single good is allocated to more than $3 \log  m$ agents and the total number of goods assigned, with copies, is at most $m$. That is, the characteristic vector $\chi^\calA$ of $\calA$ satisfies $\ellinfty{\chi^\calA} \leq 3 \log  m$ and $\ellone{\chi^\calA}\leq m$. 
\end{theorem}
Theorem \ref{thm:monotone_goods} is obtained via a direct application of the probabilitic method. We note below that the desired bounds are satisfied, with positive probability, by a random multi-allocation $\calR = (R_1, \ldots, R_n)$ in which, for each agent $i$, the bundle $R_i$ is chosen uniformly at random from among the subsets in $\calM^i=(M^i_1, \ldots, M^i_n)$. Hence, we obtain that there necessarily exists a multi-allocation that upholds the stated $\ell_1$ and $\ell_\infty$ bounds.  

\paragraph{Random Sampling.} Recall that $\calM^i=(M^i_1,\dots,M^i_n)$  denotes an MMS-inducing partition for agent $i$ (Definition \ref{def:mms}). We select, independently for each $i \in [n]$, a bundle $R_i$ uniformly at random from among $\{M^i_1,\dots,M^i_n\}$, i.e., $\Pr \{ R_i = M^i_j \} = 1/n$, for each index $j \in [n]$. 

Further, considering the (random) multi-allocation $\calR=(R_1,\ldots,R_n)$, we obtain relevant probabilistic guarantees (in Lemma \ref{lem:linftynorm} below) for its characteristic vector, $\chi^\calR \in \mathbb{Z}_+^m$. Specifically, let $G_1$ denote the event that $\ellinfty{\chi^\calR} \leq 3 \log m$ and $G_2$ denote the event that $\ellone{\chi^\calR} \leq m$. The following lemma provides lower bounds on the probabilities of $G_1$ and $G_2$. 

\begin{lemma}
\label{lem:linftynorm} 
For the (random) characteristic vector $\chi^\calR$ we have
\begin{enumerate}
\item[(i)] The expected value $\expval{\chi^\calR_g} = 1$, for each component (good) $g \in [m]$. 
\item[(ii)] $\Pr\{G_1^c\} \leq \frac{1}{m^2}.$
\item[(iii)] $\Pr\{G_2^c\} \leq 1 - \frac{1}{m+1}.$
\end{enumerate}
\end{lemma}
\begin{proof} 
We first prove part (i) of the lemma. Towards this, for each $i \in [n]$ and each $g \in [m]$, let $\mathbb{1}_{\{g \in R_i\}}$ be the indicator random variable for the event: $g \in R_i$. That is, the random variable $\mathbb{1}_{\{g \in R_i\}}$ is equal to one if $g \in R_i$, it is zero otherwise. 
 
Since $\calM^i=(M^i_1,\ldots,M^i_n)$ is an $n$-partition of $[m]$, each good $g \in [m]$ belongs to exactly one subset $M^i_j$. Further, given that the subset $R_i$ is chosen uniformly at random from among $\{M^i_1,\ldots,M^i_n\}$, the following bound holds for each $i \in [n]$ and $g \in [m]$
\begin{align}
  \expval{\mathbb{1}_{\{g \in R_i\}}} = \Pr\{g \in R_i\} = \frac{1}{n}  \label{eq:uar}
\end{align}  

Recall that $\chi^\calR_g$ is equal to the number of copies of any good $g$ assigned among the agents under the multi-allocation $\calR$, i.e., $\chi^\calR_g = \sum_{i=1}^n \mathbb{1}_{\{g \in R_i\}}$. Using this equation we obtain part (i) of the lemma: 
\begin{align}
\expval{\chi^\calR_g} = \sum_{i=1}^n \expval{\mathbb{1}_{\{g \in R_i\}}}  \underset{\text{via (\ref{eq:uar})}}{=} n \  \frac{1}{n} = 1 \label{eqn:part-one-lemma-1} 
\end{align}

For part (ii) of the lemma, note that the random (independent) selection of $R_i$s ensures that, for each fixed $g \in [m]$, the random variables $\mathbb{1}_{\{g \in R_i\}}$, across $i$s, are independent.\footnote{However, for any fixed $i$, we do not have independence between $\mathbb{1}_{\{g \in R_i\}}$ and $\mathbb{1}_{\{g' \in R_i\}}$, for goods $g,g' \in [m]$.} This observation implies that, for each fixed $g \in [m]$, the count $\chi^\calR_g = \sum_{i=1}^n \mathbb{1}_{\{g \in R_i\}}$ is a sum of independent Bernoulli random variables. Hence, via Chernoff bound~\cite[Theorem 4.4]{mitzenmacher2017probability} and for any $t \geq 6 \expval{\chi^\calR_g}$, we have $\Pr \left\{ \chi^\calR_g \geq t \right\} \leq \frac{1}{2^t}$. We instantiate this bound with $t = \max\{ 6, 3 \log m\} \geq 6 \expval{\chi^\calR_g}$; in particular, for $m \geq 4$ we have $t = 3 \log m$. Hence, we obtain $\Pr \left\{ \chi^\calR_g \geq 3 \log m \right\} \leq \frac{1}{m^3}$, for each $g \in [m]$. Let $G^c_{1,g}$ denote the event $\{\chi^\calR_g \geq 3 \log m\}$. Indeed, $\Pr \left\{ G^c_{1,g} \right\} \leq \frac{1}{m^3}$, for each $g \in [m]$, and the event $G^c_1$ (from the lemma statement) satisfies $G^c_1 = \cup_{g \in [m]} \ G^c_{1,g}$. Hence, applying the union bound gives us part (ii):
\begin{align*}
    \Pr\left\{ G^c_1 \right\} =  \Pr\left\{ \cup_{g \in [m]} \ G^c_{1,g} \right\} \leq \sum_{g =1}^m \Pr \left\{ G^c_{1,g} \right\} \leq m \ \frac{1}{m^3} = \frac{1}{m^2}. 
\end{align*}

Finally, for part (iii), observe that $ \expval{\ellone{\chi^\calR}} = \expval{\sum_{g=1}^m \chi^\calR_g} = \sum_{g=1}^m \expval{\chi^\calR_g} = m$; the last equality follows from equation (\ref{eqn:part-one-lemma-1}). Further, for the nonnegative random variable $\ellone{\chi^\calR}$, Markov's inequality gives us $\Pr\left\{ 
 \ellone{\chi^\calR} \geq (m+1) \right\} \leq \frac{\expval{\ellone{\chi^\calR}}}{m+1} = \frac{m}{m+1}$. Since $\ellone{\chi^\calR}$ is an integer-valued random variable, $\Pr\left\{ \ellone{\chi^\calR} > m \right\} = \Pr\left\{ \ellone{\chi^\calR} \geq (m+1) \right\}$. These observations lead to the bound stated in part (iii):
 \begin{align*}
     \Pr\left\{ G_2^c \right\} = \Pr\left\{ \ellone{\chi^\calR} > m \right\} = \Pr\left\{ \ellone{\chi^\calR} \geq (m+1) \right\} \leq \frac{m}{m+1} = 1 - \frac{1}{m+1}.
 \end{align*}
This completes the proof of the lemma. 
\end{proof}

With the above lemma in hand, we now prove \Cref{thm:monotone_goods}.

\begin{proof}[Proof of \Cref{thm:monotone_goods}] As mentioned previously, $G_1$ denotes the event $\ellinfty{\chi^\calR} \leq 3 \log m$ and $G_2$ denotes $\ellone{\chi^\calR} \leq m$. Lemma~\ref{lem:linftynorm} implies that $G_1$ and $G_2$ hold together with positive probability
\begin{align*}
\prob{G_1 \cap G_2} & = 1 - \prob{G_1^c \cup G_2^c} \\ 
& \geq 1 - (\prob{G_1^c} +\prob{G_1^c}) \tag{Union Bound} \\
& \geq 1-\left(\frac{1}{m^2} + 1 - \frac{1}{m+1}\right) \tag{\Cref{lem:linftynorm}} \\
\\ & =\frac{1}{m+1} - \frac{1}{m^2} > 0
\end{align*}

Hence, by the definitions of $G_1$ and $G_2$, we get that the random multi-allocation $\calR=(R_1, \ldots, R_n)$ satisfies the stated $\ell_\infty$ and $\ell_1$ bounds. Moreover, for each $i$, the bundle $R_i$ is sampled from among the subsets that form the MMS-inducing partition $\calM^i$. Hence, $v_i(R_i) \geq \mu_i$, for every $i \in [n]$, ensuring that $\calR$ is always an MMS multi-allocation. Overall, these observations imply that, as claimed in the theorem, there always exists an MMS multi-allocation $\calA$ with the properties that $\ellinfty{\chi^\calA} \leq 3 \log  m$ and $\ellone{\chi^\calA}\leq m$. The theorem stands proved. 
\end{proof}

\subsection{Identically Ordered Valuations}\label{subsec:additive-ordered}
Recall that valuations in a fair division instance are said to be identically ordered if there exists an indexing $\{g_1, \ldots g_m\}$ of the goods such that for each pair of goods $g_s, g_t$, with index $s < t$, and all agents $i \in [n]$, the inequality $v_i(S + {g_s }) \geq  v_i(S + {g_t })$ holds for each subset $S \subset [m]$ that does not contain $g_s$ and $g_t$. Note that additive ordered valuations are identically ordered. The following example highlights that identically ordered valuations are not confined to be subadditive, or even superadditive. 

\begin{example} \label{ex:sqrt-ordered}
 Let $g_1, \dots, g_m$ be a fixed indexing of the set $[m]$ of goods, and, for each agent $i$, let $w_i : [m] \to \mathbb{R}_+$ be weights on the goods that satisfy $w_i(g_1) \geq w_i(g_2) \geq \dots \geq w_i(g_m)$. Note that, for all the agents, the weights respect the common indexing. For each $i \in [n]$, let $f_i: \mathbb{R}_+ \mapsto \mathbb{R}_+$ be a monotone nondecreasing function. Define valuation $v_i(S) \coloneqq f_i \left( \sum_{g\in S} w_i(g) \right)$, for each subset $S \subseteq [m]$ and agent $i$. Since the functions $f_i$s are monotone nondecreasing, the valuations $v_i$s are identically ordered: for all agents $i$, the following inequality holds for each pair of goods $g_s, g_t$, with index $s<t$, and each subset $S$ that does not contain $g_s$ and $g_t$: 
 \begin{align*}
v_i(S + g_s) = f_i \left( \sum_{g \in S} w_i(g)  \ +  \ w_i(g_s) \right) \geq  f_i \left( \sum_{g \in S} w_i(g)  \ +  \ w_i(g_t) \right) = v_i(S + g_t).
\end{align*}
With $f_i$s as the identify function, we obtain additive ordered valuations. Setting $f_i(w) = \sqrt{w}$ gives us subadditive valuations. On the other hand, with $f_i(w) = \exp(w)$, we obtain superadditive valuations. In all of these cases, our result for identically ordered valuations holds. 
\end{example}
 

Theorem \ref{thm:additive_goods} (stated below) provides our upper bounds on the supply adjustments for MMS under identically ordered valuations.

\begin{theorem}
\label{thm:additive_goods} Every fair division instance $\langle [n], [m], \{v_i\}_{i=1}^n \rangle$ with identically ordered valuations admits an MMS multi-allocation $\calA = (A_1, \ldots A_n)$ in which no single good is assigned to more than $12 \sqrt{\log m}$ agents and the total number of goods assigned, with copies, is at most $m + \frac{6m  \sqrt{\log m} }{\sqrt{n}}$. That is, the characteristic vector $\chi^\calA$ of $\calA$ satisfies 
\begin{align*}
    \ellinfty{\chi^\calA} \leq 12 \sqrt{\log m}  \qquad { \text{and} } \qquad \ellone{\chi^\calA} \leq m + \frac{6m  \sqrt{\log m} }{\sqrt{n}}. 
\end{align*}
\end{theorem}

\paragraph{Random Sampling.} As in the case of monotone valuations, we sample an MMS multi-allocation, $\calR=(R_1, \ldots, R_n)$, uniformly at random from among the MMS-inducing bundles of each agent, albeit with a key modification. For each agent $i\in [n]$, let $\calM^i=(M^i_1,\dots,M^i_n)$ be an MMS-inducing partition for $i$ in the given instance $\langle [n], [m], \{v_i\}_{i=1}^n \rangle$. In each $\calM^i$, index the subsets such that $|M^i_1| \leq |M^i_2| \leq \dots \leq |M^i_n|$. Let $s_i \in [n]$ be the largest index with the property that $|M^i_{s_i}| \leq \frac{2m}{n}$. Hence, for each index $t \in \{1, 2, \ldots, s_i\}$, we have $|M^i_t| \leq \frac{2m}{n}$. Also, since $M^i_j$s partition $[m]$, at most $n/2$ of these subsets can have cardinality more than $\frac{2m}{n}$, i.e., index $s_i \geq n/2$.  

We select, independently for each $i \in [n]$, a bundle $R_i$ uniformly at random from among $\{M^i_1,\dots,M^i_{s_i} \}$, i.e., $\Pr \{ R_i = M^i_t \} = \frac{1}{s_i}$, for each index $t \in \{1, 2, \ldots, s_i\}$. For the sampled multi-allocation $\calR=(R_1, \ldots, R_n)$, we define the event $G_1 \coloneqq \left\{\ellone{\chi^\calR} \leq m + \frac{6m  \sqrt{\log m} }{\sqrt{n}}\right\}$, which corresponds to our claimed $\ell_1$ bound. We first show that $G_1$ holds with high probability.

\begin{lemma}
\label{lem:l1_additive}
    $\prob{G_1^c} \leq \frac{2}{m^3}$.
\end{lemma}

\begin{proof} 
Note that $\ellone{\chi^\calR}=\sum_{i=1}^n |R_i|$ is a sum of independent random variables, $|R_i|$, each in the range $\left[0,\frac{2m}{n}\right]$. Moreover, for each $i$ we have 
\begin{align}
    \expval{|R_i|} = \frac{1}{s_i} \sum_{j=1}^{s_i} \left|M^i_j\right| \leq \frac{1}{n} \sum_{j=1}^{n} \left|M^i_j\right|=\frac{m}{n} \label{ineq:expRone}
\end{align}
Here, the first inequality follows from the fact that  $M^i_1, M^i_2, \ldots, M^i_{s_i}$ are the $s_i$ subsets of smallest cardinality in $\calM^i=(M^i_1, \ldots, M^i_n)$. Hence, their average cardinality, $\frac{1}{s_i} \sum_{j=1}^{s_i} |M^i_j|$, is at most the overall average $\frac{1}{n} \sum_{j=1}^{n} |M^i_j|$. 

Inequality (\ref{ineq:expRone}) gives us $\expval{\ellone{\chi^\calR}} = \sum_{i =1}^n \expval{|R_i|} \leq \sum_{i=1}^n \frac{m}{n}=m$. 

Further, applying Hoeffding's inequality with $\delta \coloneqq \frac{6m\sqrt{\log m}}{\sqrt{n}}$, we obtain\footnote{Here, the term $\frac{4m^2}{n^2}$ in the denominator of the exponent follows from the range of random variable, $|R_i| \in \left[0,\frac{2m}{n}\right]$.}

\begin{align*} \displaystyle \prob{G_i^c}=\prob{\ellone{\chi^\calR} > m + \delta} & \leq 2 \exp{\left(- \frac{\delta^2}{3 \ n \ \frac{4m^2}{n^2}}\right)} \\ & = 2 \exp{\left(\frac{-3 m^2 \log m}{m^2}\right)} \\ & \leq \frac{2}{m^3}. \end{align*}
The lemma stands proved. 
\end{proof}

Let $g_1, \dots, g_m$ be an indexing of the set $[m]$ of goods that satisfies to property laid out in \Cref{def:identically-ordered}. We will now define a probabilistic event $G_2$ considering the \textit{dyadic prefixes} of $\{g_1, g_2, \ldots, g_m\}$ and the random multi-allocation $\calR$. For each $t\in\{0,1,\ldots,\log m\}$, define the prefix $P_t \coloneqq \{g_1,g_2,\ldots,g_{2^t} \}$ to be the set of the $2^t$ lowest-indexed goods (recall that these have the highest marginal values). In addition, let $\calP$ be the collection of all such prefixes, $\mathcal{P} \coloneqq \{P_0,P_1,\dots,P_{\log m}\}$. 

For each $P\in \mathcal{P}$, we define $\chi^\calR_P$ to be the characteristic vector $\chi^\calR$ restricted to entries corresponding to goods in $P$. In particular, $\chi^\calR_P$ is a $|P|$-dimensional vector, and $\ellone{\chi^\calR_P}$ denotes the total number of goods (with copies) from subset $P$ that are assigned in $\calR$, i.e., $\ellone{\chi^\calR_P} = \sum_{g \in P} \chi^\calR_g$. 

The event $G_2$ bounds these assignment numbers for all the prefixes $P$. Formally, 
\begin{align*}
\text{Event } G_2 \coloneqq \left\{ \ellone{\chi^\calR_P}\leq 6 \sqrt{\log m}\ |P| \ \ \text{ for every } P\in\mathcal{P} \right\}.
\end{align*}

\begin{lemma}
\label{lem:dyadic-prefixes} $\prob{G_2}\geq \frac{1}{2}$.
\end{lemma}

\begin{proof}
    Fix any prefix subset $P\in \mathcal{P}$ and write $r^i_P$ to denote the characteristic vector $\rmchar(R_i)$ restricted to the components/goods in subset $P$. That is, for each $g \in P$, the $g$th component of vector $r^i_P$ is equal to one if $g \in R_i$, it is zero otherwise. 
    Note that $\chi^\calR_P=\sum_{i=1}^n r^i_P$. Therefore, 
        \begin{align}
        \expval{\left\| \chi^\calR_P \right\|_2^2} 
        &= \sum_{i=1}^n \expval{\left\| r^i_P \right\|_2^2} + \sum_{i\neq j} \expval{\left\langle r^i_P,r^j_P\right\rangle} \nonumber \\
        &= \sum_{i=1}^n \expval{\left\|r^i_P\right\|_2^2} + \sum_{i\neq j} \left\langle \expval{r^i_P},\expval{r^j_P} \right\rangle \label{eq:ell2}
    \end{align} 
   The last equality follows from the fact that $r^i_P$ and $r^j_P$ are independent for $i \neq j$. Now, for each $i\in [n]$, since $r_P^i$ is a binary (characteristic) vector, we have 
   \begin{align*}
       \expval{\left\|r^i_P\right\|_2^2}
       =\expval{\left|R_i\cap P\right|}
       =\frac{1}{s_i} \sum_{k=1}^{s_i} \left|M^i_k\cap P\right|
       \leq \frac{|P|}{s_i}
       \leq \frac{2|P|}{n} \tag{since $s_i \geq n/2$}
   \end{align*}
    Also, for each agent $i$, since $\expval{r^i_P}$ is a vector with all entries either $0$ or $\frac{1}{s_i}$. Hence,  
    \begin{align*}
        \left\langle \expval{r^i_P},\expval{r^j_P} \right\rangle \leq \frac{|P|}{s_i\  s_j}\leq\frac{4|P|}{n^2} \tag{since each $s_i \geq n/2$}
    \end{align*}
    Therefore, equation (\ref{eq:ell2}) reduces to 
    \begin{align}
    \expval{\left\|\chi^\calR_P\right\|_2^2}
    \leq \sum_{i=1}^n \frac{2|P|}{n} + \sum_{i\neq j} \frac{4|P|}{n^2} 
    =  2|P|+\binom{n}{2} \frac{4|P|}{n^2} 
    \leq 4|P| \label{ineq:ell2bnd}
    \end{align}
Using (\ref{ineq:ell2bnd}), we next upper bound the expected value and standard deviation of random variable $\ellone{\chi^\calR_P}$. Towards this, first recall that, for any $|P|$-dimensional vector $x \in \mathbb{R}^{|P|}$, the Cauchy-Schwartz inequality gives us $\|x \|_1 \leq \sqrt{|P|} \  \| x \|_2$. In particular, $\ellone{\chi^\calR_P}^2 \leq |P| \left\|\chi^\calR_P\right\|_2^2$ and, hence, 
    \begin{align}
        \expval{\left\|\chi^\calR_P\right\|_1^2} \leq |P| \ \expval{\left\|\chi^\calR_P\right\|_2^2} \underset{\text{via (\ref{ineq:ell2bnd})}}{\leq} 4 |P|^2 \label{ineq:ell1ell2}
    \end{align}
The expected value of $\ellone{\chi^\calR_P}$ satisfies 
\begin{align}
\expval{\ellone{\chi^\calR_P}} &\leq \sqrt{\expval{\ellone{\chi^\calR_P}^2}} \tag{Jensen's inequality} \\
& \leq \sqrt{ 4 |P|^2 } \tag{via (\ref{ineq:ell1ell2})} \\
& = 2 \ |P| \label{ineq:ell1bound}
\end{align}
       Next, for the standard deviation $\sigma\left[\ellone{\chi^\calR_P} \right]$ we have 
    \begin{align}
        \sigma\left[ \ellone{\chi^\calR_P} \right]
    =\sqrt{ \mathrm{Var} \left[ \ellone{\chi^\calR_P} \right]}
    \leq \sqrt{\expval{\ellone{\chi^\calR_P}^2}} \underset{\text{via (\ref{ineq:ell1ell2})}}{\leq} 2\ |P| \label{ineq:stddev}
    \end{align}

    Let $\theta$ denote the expected value of $\ellone{\chi^\calR_P}$ and $\sigma$ denote its standard deviation. Inequalities (\ref{ineq:ell1bound}) and (\ref{ineq:stddev}) ensure that $\theta \leq 2|P|$ and $\sigma \leq 2 |P|$. Further, Chebyshev's inequality gives us $\prob{ \left| \ellone{\chi^\calR_P} - \theta \right| > k \ \sigma  } \leq \frac{1}{k^2}$ for any $k \geq 1$. Instantiating the inequality with $k = 2\sqrt{\log m}$, we obtain  
    \begin{align}
        \prob{ \ellone{\chi^\calR_P} > 6 \sqrt{\log m} \ |P| } & \leq \prob{ \ellone{\chi^\calR_P} > 4  \sqrt{\log m} \ |P| + 2|P| } \nonumber \\  
        & \leq \prob{\ellone{\chi^\calR_P} > 2  \sqrt{\log m} \ \sigma + \theta }  \tag{$\theta, \sigma \leq 2|P|$} \\ 
        & \leq \prob{\ellone{\chi^\calR_P} > k \ \sigma + \theta }  \tag{$k = 2\sqrt{\log m}$} \\ 
        & \leq \prob{ \left| \ellone{\chi^\calR_P} - \theta \right| > k \ \sigma  } \nonumber \\
        & \leq \frac{1}{4 \log m} 
    \end{align}
 Therefore, for each dyadic prefix subset $P \in \calP$, it holds that 
 \begin{align}
     \prob{\ellone{\chi^\calR_P} > 6 \sqrt{\log m} \ |P| } \leq \frac{1}{4\log m} \label{ineq:eachP}
 \end{align} 
 Recall that event $G_2 \coloneqq \left\{ \ellone{\chi^\calR_P}\leq 6 \sqrt{\log m}\ |P| \ \ \text{ for every } P\in\mathcal{P} \right\}$. To upper bound the complement of $G_2$, we apply union bound considering all the $(\log m + 1)$ dyadic prefixes $P\in\mathcal{P}$. Specifically,   
  \begin{align*} \prob{G^c_2} & \leq \sum_{t = 0}^{\log m} \prob{\ellone{\chi^\calR_{P_t}} > 6 \sqrt{\log m}\cdot |P_t|} 
  \tag{Union Bound} \\ 
  & \leq \frac{\log m + 1}{4\log m} \tag{via (\ref{ineq:eachP})} \\
  & < \frac{1}{2}.
  \end{align*}
  Hence, $\prob{G_2} \geq 1/2$, and the lemma stands proved. 
\end{proof}

\begin{corollary} \label{cor:existence-of-b}
  Events $G_1$ and $G_2$---as defined for Lemmas \ref{lem:l1_additive} and \ref{lem:dyadic-prefixes}---hold together with positive probability, $\prob{G_1\cap G_2}>0$. Consequently, every fair division instance with additive ordered valuations admits an MMS multi-allocation $\calB=(B_1,\dots,B_n)$ with the properties that 
  \begin{itemize}
        \item[(i)] $\ellone{\chi^\calB} \leq m + \frac{6m  \sqrt{\log m} }{\sqrt{n}}$
        \item[(ii)] $\ellone{\chi^\calB_P}\leq 6 \sqrt{\log m}\ |P|$, for every dyadic prefix $P\in\mathcal{P}$.
    \end{itemize}
\end{corollary}

\begin{proof} 
Lemmas \ref{lem:l1_additive} and \ref{lem:dyadic-prefixes} give us 
\begin{align*} \displaystyle \mathbb{P}\{G_1 \cap G_2\} & = 1 - \mathbb{P}\{G_1^c \cup G_2^c\} \\ 
& = 1 - \left(\mathbb{P}\{G_1^c\}+ \mathbb{P}\{G_2^c\}\right) & (\text{Union Bound})\\ 
& = 1 - \left(\frac{1}{2} + \frac{2}{m^3}\right) & (\text{\Cref{lem:l1_additive} and \ref{lem:dyadic-prefixes}}) \\ 
& > 0.
\end{align*} 
Since this probability is strictly positive, the random allocation $\calR=(R_1, \ldots, R_n)$ satisfies the desired properties with positive probability. Also, recall that each $R_i$ is drawn considering the MMS-inducing partition $\calM^i$ and, hence, $\calR$ is always an MMS multi-allocation. Overall, we get that there necessarily exists an MMS multi-allocation $\calB$ that satisfies properties (i) and (ii) in the corollary. This completes the proof. 
\end{proof}

Given a multi-allocation $\calB=(B_1,\ldots,B_n)$ as guaranteed by Corollary \ref{cor:existence-of-b}, the following lemma shows that we can construct a multi-allocation $\calA=(A_1,\ldots,A_n)$ that upholds Theorem \ref{thm:additive_goods}.

Recall that, under identically ordered valuations, the goods are indexed $g_1, \ldots, g_m$ in decreasing order of marginal values. That is,  for each pair of goods $g_s$ and $g_t$, with index $s<t$, we have $v_i(S + g_s) \geq v_i(S + g_t)$ for every agent $i$ and each subset $S$ that does not contain $g_s$ and $g_t$. Also, the collection $\mathcal{P}=\{P_0,P_1,\dots,P_{\log m}\}$ contains all the dyadic prefixes, $P_t=\{g_1,g_2,\ldots,g_{2^t}\}$ with $t\in\{0,1,\dots,\log m\}$. We will construct the desired multi-allocation $\calA$ from $\calB$ by systematically replacing copies of a good $\ell$ with lower marginal value, by those of a good $h$ with higher marginal value; see Algorithm \ref{alg:copy-redistribution} for details.
    
\begin{lemma}\label{lem:copy-redistribution}
        In any fair division instance with identically ordered valuations, let $\calB$ be an MMS multi-allocation that satisfies the properties stated in Corollary \ref{cor:existence-of-b}. Then, from $\calB$, we can construct an MMS multi-allocation $\calA$ in the given instance such that $\ellinfty{\chi^\calA} \leq 12 \sqrt{\log m}$ and $\ellone{\chi^\calA}=\ellone{\chi^\calB}$.
\end{lemma}

\begin{algorithm}
\caption{Copy Redistribution for Identically Ordered Valuations}\label{alg:copy-redistribution} 
\begin{algorithmic}[1]
        \REQUIRE An MMS multi-allocation $\calB$ satisfying the properties of Corollary \ref{cor:existence-of-b}\\ 
        \ENSURE A multi-allocation $\calA$
        \STATE Initialize $\calA = \calB$. 
        \WHILE{$\ellinfty{\chi^\calA} > 12 \sqrt{\log m}$}
        \STATE Let $g_t \in [m]$ be the good with the lowest index $t$ that satisfies $\chi^\calA_{g_t} > 12 \sqrt{\log m}$. \label{line:lowgood}
        \STATE Let $g_s \in [m]$ be any good of lower index $s < t$ (equivalently, higher marginal value) with $\chi^\calA_{g_s} < 12 \sqrt{\log m}$. \label{line:highgood}
        \COMMENT{We show that such a good $g_s$ always exists.}
        \STATE Pick an agent $i$ such that good $g_t \in A_i$ and $g_s \notin A_i$. 
        \COMMENT{Since $\chi^\calA_{g_s} < \chi^\calA_{{g_t}}$, such an agent $i$ necessarily exists.}
        \STATE Update $A_i \gets \left( A_i\setminus\{g_t \} \right) \cup\{g_s \}$. \\ \COMMENT{This iteration of the while loop reduces the number of copies of good $g_t$ by commensurately increasing the number of copies of good $g_s$, where $s < t$.}
        \ENDWHILE
    \RETURN $\calA$
\end{algorithmic}
\end{algorithm}

\begin{proof}
We will first show that the while-loop in Algorithm \ref{alg:copy-redistribution} successfully executes for each maintained multi-allocation $\calA$ with $\ellinfty{\chi^\calA} > 12 \sqrt{\log m}$. In particular, we will prove that if $\ellinfty{\chi^\calA} > 12 \sqrt{\log m}$, for a maintained $\calA$, then the good $g_s$ sought in Line \ref{line:highgood} necessarily exists. Also, note that the while-loop must terminate after polynomially many iterations, since the following potential strictly decreases in each iteration: $\sum_{g \in [m]} \ \max\{ 0, \chi^\calA_g - 12 \sqrt{\log m}\}$. Hence, for the returned multi-allocation $\calA$ (obtained after the while-loop terminates), we have $\ellinfty{\chi^\calA} \leq 12 \sqrt{\log m}$.

Consider any iteration of the while-loop. Let $\calA$ be the current multi-allocation (with $\ellinfty{\chi^\calA} > 12 \sqrt{\log m}$) and $g_t$ be the good selected in Line \ref{line:lowgood}. Assume here, towards a contradiction, that the good $g_s$ desired in Line \ref{line:highgood} does not exist. That is, for each good $k \in L_t$, we have $\chi^\calA_k \geq 12 \sqrt{\log m}$; here, $L_t$ denotes the subset of all the goods with index at most $t$. 
By definition, $g_t \in L_t$. Also, note that, the selection criterion in Line \ref{line:lowgood} ensures that all the goods' reassignments that have happened before the current iteration must have been between goods in $L_t$. That is, each pair of goods $g_{t'}$ and $g_{s'}$ considered in any previous iteration satisfy $g_{t'}, g_{s'} \in L_t$. This observation implies that, in the current multi-allocation $\calA$ and the initial multi-allocation $\calB$, the total number of assignments (with copies) of the goods in $L_t$ are equal:  
\begin{align}
    \sum_{k \in L_t} \chi^\calA_k = \sum_{k \in L_t} \chi^\calB_k \label{eq:samesum}
\end{align} 
Since $\chi^\calA_k \geq 12 \sqrt{\log m}$ for each $k \in L_t$, equation (\ref{eq:samesum}) reduces to $\sum_{k \in L_t} \chi^\calB_k \geq  \ 12  \sqrt{\log m} |L_t| $. Select $\ell \in \{0,1,\ldots, \log m\}$ such that $2^\ell < t \leq 2^{\ell+1}$, i.e., $g_t \in P_{\ell +1} \setminus P_\ell$. This containment implies $P_\ell \subset L_t \subseteq P_{\ell+1}$ and, hence, $|L_t| > |P_\ell| = 2^\ell$. Therefore, the above-mentioned bound extends to 
\begin{align}
    \sum_{k \in L_t} \chi^\calB_k \geq   12 \sqrt{\log m} \ |L_t| >  12 \sqrt{\log m} \ \  2^\ell =  6 \sqrt{\log m} \ \  2^{\ell+1} = 6 \sqrt{\log m} |P_{\ell+1}| \label{ineq:prefixpack}
\end{align}
Since $L_t \subseteq P_{\ell+1}$, inequality (\ref{ineq:prefixpack}) implies $\sum_{k \in P_{\ell+1}} \chi^\calB_k >  6 \sqrt{\log m} |P_{\ell+1}| $. This bound, however, contradicts the fact that multi-allocation $\calB$ upholds property (ii) in Corollary \ref{cor:existence-of-b} for $P_{\ell + 1} \in \calP$. 

Hence, by way of contradiction, we get that each iteration of the while-loop executes successfully and, at termination, we obtain a multi-allocation $\calA$ with $\ellinfty{\chi^\calA} \leq 12 \sqrt{\log m}$. That is, the returned multi-allocation satisfies the stated $\ell_\infty$ bound. 

Further, note that, in each iteration of the while loop, the cardinality of the bundle $A_i$ is maintained -- we add a lower-indexed good, $g_s$, to it and remove a higher-indexed one, $g_t$. Since the valuations are identically ordered, in any such replacement, the updated value of agent $i$'s bundle, $v_i((A_i \setminus \{g_t \}) \cup \{g_s\})$, is at least its original value $v_i(A_i) = v_i((A_i \setminus \{g_t\}) \cup \{g_t\})$. Hence, for each agent $i$, the following inequalities continue to hold: $v_i(A_i) \geq v_i(B_i)$ and $|A_i| = |B_i|$. Therefore, for the returned multi-allocation $\calA$, the stated $\ell_1$ bound holds: $\ellone{\chi^\calA} = \sum_{i=1}^n |A_i| = \sum_{i=1}^n |B_i| = \ellone{\chi^\calB}$. Finally, the facts that $\calB$ is an MMS multi-allocation and $v_i(A_i) \geq v_i(B_i)$, for each $i$, imply that $\calA$ is also MMS. 
Overall, we have that the returned allocation $\calA$ is MMS, and it satisfies the stated $\ell_1$ and $\ell_\infty$ bounds. The lemma stands proved.     
\end{proof}

\begin{proof}[Proof of \Cref{thm:additive_goods}]
    Let $\calB=(B_1,\dots, B_n)$ be the MMS multi-allocation whose existence is guaranteed by \Cref{cor:existence-of-b}. Starting with $\calB$ and applying \Cref{lem:copy-redistribution} (\Cref{alg:copy-redistribution}), we can obtain an MMS multi-allocation $\calA=(A_1,\dots,A_n)$ with the properties that $\ellinfty{\chi^\calA} \leq 12 \sqrt{\log m}$ and $\ellone{\chi^\calA}=\ellone{\chi^\calB} \leq m+\frac{6m\sqrt{\log m}}{\sqrt{n}}$. The guaranteed existence of such an MMS multi-allocation establishes the theorem. 
\end{proof}

\paragraph{Remark.} It is relevant to note that prior works on MMS, for additive valuations, assume that one can restrict attention, without loss of generality, to instances with valuations that identically ordered in addition to being additive. This follows from a reduction of Bouveret and Lema\^{i}tre \cite{10.5555/2615731.2617458} via which, and for any instance $\calI$ with additive valuations, one can construct an instance $\calI'$ with additive and identically ordered valuations, such that if $\calI'$ admits an MMS allocation $\calA'$ then, in fact, $\calI$ also admits an MMS allocation $\calA$.\footnote{See \cite{barman2020approximation} for an approximation-preserving version of this result.} The reduction works by deriving from $\calA'$ a {\it picking sequence} over the agents and then showing that greedily assigning goods in $\calI$, via this sequence, leads to an MMS allocation $\calA$.  

Such a reduction, however, is not immediate in the case of multi-allocations. In the current context of allocating copies of goods, the difficulty stems from the fact that while constructing $\calA$ (via the picking sequence mentioned above) one might be forced to assign multiple copies of the same good to an agent $i$. Hence, it remains open whether the guarantee obtained (in Theorem \ref{thm:additive_goods}) for additive ordered valuations extends to all additive valuations. 

At the same time, we note that Theorem \ref{thm:additive_goods} generalizes to additive valuations if one were to relinquish the requirement that each agent $i$'s bundle, $A_i$, must be a subset of $[m]$. That is, in contrast to the rest of the paper, if $A_i$ is allowed to be a multiset, in which each extra copy of any good $g$ fetches an additional value of $v_i(g)$, then the reduction from \cite{10.5555/2615731.2617458} applies and the guarantee given in Theorem \ref{thm:additive_goods} extends to all additive valuations.

\subsection{Additive Valuations}
This section focuses on fair division of goods under additive valuations. 

\begin{theorem} 
\label{thm:additive_3_3m}
        Every fair division instance $\langle [n], [m], \{v_i\}_{i=1}^n \rangle$  with additive valuations admits an MMS multi-allocation $\calA$ in which no single good is assigned to more than $2$ agents and the total number of goods assigned, with copies, is at most $2m$. That is, the characteristic vector $\chi^\calA$ of $\calA$ satisfies $\ellinfty{\chi^\calA} \leq 2$ and $\ellone{\chi^\calA} \leq 2m$.
\end{theorem}

\begin{proof}
Let $\calI=\langle [n],[m],\{v_i\}_{i=1}^n\rangle$ be the given instance. We construct an auxiliary instance $\widetilde{\calI}$ (detailed next) and consider a constrained fair division problem over it. The instance $\widetilde{\calI}=\langle [n], [m] \times[2],\{\widetilde{v}_i\}_{i=1}^n\rangle$ is constructed as follows: The set of agents is unchanged, whereas each good $g\in [m]$ is replicated twice as $(g,1)$ and $(g,2)$  in $\widetilde{\calI}$. These two goods are referred to as \textit{copies} of $g$ in $\widetilde{\calI}$. For each agent $i \in [n]$, the valuation function $\widetilde{v}_i$ in $\widetilde{\calI}$ is additive and obtained by setting the values of the two copies equal to the value of the underlying good: $\widetilde{v}_i((g,1)) = \widetilde{v}_i((g,2)) = v_i(g)$ for every good $g \in [m]$. Hence, for any subset $S \subseteq [m] \times [2]$, we have $\widetilde{v}_i (S) = \sum_{(g,k) \in S} \ v_i(g) $.    

In instance $\widetilde{\calI}$, we focus on maximin shares under cardinality constraints to impose the requirement that each agent receives at most one copy of each good. Specifically, in $\widetilde{\calI}$,  a subset of goods $S \subseteq [m] \times [2]$ is said to be {\it feasible} if $\big| S \cap \{(g,1), (g,2) \} \big| \leq 1$ for every good $g$. Note that for any feasible subset $S \subseteq [m] \times [2]$, the valuations under $\widetilde{v}_i$ and $v_i$ match, $\widetilde{v}_i(S) = v_i \left( S \odot [m] \right)$, where the projected set $S \odot [m] \coloneqq \left\{ g \in [m] \ \mid \ (g,1) \in S \text{ or } (g,2) \in S \right\}$. 

We consider maximin shares, $\widetilde{\mu}_i$, in instance $\widetilde{\calI}$ under these cardinality constraints
\begin{align}
    \widetilde{\mu_i} \coloneqq \max_{\substack{(X_1, \ldots, X_n) \in \Pi_n([m] \times [2]): \\ \text{each $X_i$ is feasible}} } \ \ \min_{j \in [n]} \  \widetilde{v}_i  (X_j) \label{eqn:mmshat}
\end{align}
Maximin shares under cardinality constraints have been studied in prior works; see, e.g, \cite{hummel2022maximin} and \cite{biswas2018fair}. In particular, the work of Hummel and Hetland \cite{hummel2022maximin} shows that, under additive valuations, one can find in polynomial time an (exact) allocation $\calB = (B_1,\ldots, B_n) \in \Pi_n([m] \times [2])$ such that, for each $i \in [n]$, we have $\widetilde{v}_i(B_i) \geq \frac{1}{2} \widetilde{\mu}_i$ and $B_i$ is feasible. This algorithmic result of \cite{hummel2022maximin} implies that the constructed instance $\widetilde{\calI}$ necessarily admits such an allocation $\calB$. 

From $\calB$ we can derive a multi-allocation $\calA=(A_1, \ldots, A_n)$ in $\calI$ by setting $A_i = B_i \odot [m]$, for each $i \in [n]$. That is, $A_i$ is obtained by including good $g \in [m]$ in it iff a copy of $g$ (i.e., $(g,1)$ or $(g,2)$) is present in $B_i$. Since, in $\widetilde{\calI}$, each good $g \in [m]$ has two copies, the multi-allocation $\calA$ satisfies $\ellinfty{\chi^\calA}  = 2$. Also, since each $B_i$ is feasible, $|A_i|=|B_i|$ and, hence, $\ellone{\chi^\calA} = \sum_{i=1}^n |A_i| = \sum_{i=1}^n |B_i| = 2m$. Therefore, the multi-allocation $\calA$ satisfies the stated $\ell_1$ and $\ell_\infty$ bounds. 

It remains to show that $\calA$ is an MMS multi-allocation in the given instance $\calI$. Towards this, first note that the feasibility of $B_i$ implies $v_i(A_i) = \widetilde{v}_i(B_i) \geq \frac{1}{2} \widetilde{\mu}_i$, for each agent $i \in [n]$. Next, we will show that $\widetilde{\mu}_i \geq 2 \mu_i$, where $\mu_i$ denotes the maximin share of agent $i$ in the instance $\calI$. Together these bounds show that $\calA$ is indeed an MMS multi-allocation: $v_i(A_i) \geq \frac{1}{2} \widetilde{\mu}_i \geq \mu_i$. 

Fix any agent $i \in [n]$. To show $\widetilde{\mu}_i \geq 2 \mu_i$, we start with an MMS-inducing partition $\calM^i=(M^i_1, \ldots, M^i_n)$ for agent $i$ in instance $\calI$. Note that $v_i(M^i_j) \geq \mu_i$, for each $j \in [n]$. From $\calM^i$, construct a partition $\widetilde{\calP}=(\widetilde{P}_1,\ldots,\widetilde{P}_n)$ in $\widetilde{\calI}$ by setting $\widetilde{P}_j = \left( M^i_j \times \{1\} \right) \bigcup \left( M^i_{j+1} \times \{2 \} \right)$, for each $1 \leq j < n$, and  $\widetilde{P}_n =  \left( M^i_n \times \{1\} \right) \bigcup \left( M^i_{1} \times \{2 \} \right)$. That is, in $\widetilde{P}_j$, we include the first copy of each good $g \in M^i_j$ and the second copy of each good $g' \in M^i_{j+1}$. Since $M^i_j \cap M^i_{j+1} = \emptyset$, each $\widetilde{P}_j$ is feasible -- it contains at most one copy of any good $g\in [m]$. Moreover, 
\begin{align*}
    \widetilde{v}_i(\widetilde{P}_j) = \widetilde{v}_i \left(M^i_j \times \{1\}  \right)  +  \widetilde{v}_i \left(M^i_{j+1} \times \{2\}  \right) = v_i(M^i_j) + v_i(M^i_{j+1}) \geq 2 \mu_i.
\end{align*}
Hence, the partition $\widetilde{\calP}=(\widetilde{P}_1,\ldots,\widetilde{P}_n)$ certifies that equation (\ref{eqn:mmshat}) holds with $\widetilde{\mu}_i \geq 2 \mu_i$. As mentioned previously, this bound establishes the existence of an MMS multi-allocation $\calA$ that satisfies the stated $\ell_1$ and $\ell_\infty$ bounds. The theorem stands proved. 
\end{proof}

\subsection{Lower Bound for Goods}

We now show that there exist fair division instances $\calI$, with monotone valuations, such that each MMS multi-allocation in $\mathcal{I}$ ends up assigning at least one good to more than $\frac{\log m}{\log \log m}$ many agents. This lower bound shows that the positive result obtained in Theorem \ref{thm:monotone_goods}---for the $\ell_\infty$ bound---is essentially tight. 

\begin{theorem}
\label{thm:goods_lowerbound}
    There exists a fair division instance $\calI = \langle [n], [m], \{v_i\}_{i=1}^n \rangle$ with monotone valuations such that for every MMS multi-allocation $\calA$ in $\calI$, it holds that $\ellinfty{\chi^\calA} \geq \frac{\log m}{\log \log m}$, i.e., $\calA$ assigns some good to at least $\frac{\log m}{\log\log m}$ agents.
\end{theorem}

\begin{proof} We first describe the construction of the instance $\langle [n], [m], \{v_i\}_{i=1}^n \rangle$. For each agent $i \in [n]$ independently, we draw a partition $\mathcal{P}^i =(P^i_1, P^i_2, \ldots P^i_n)$ uniformly at random from the set of all possible $n^m$ possible partitions, i.e., $\calP^i \in_R \Pi_n([m])$. Given such a partition, we set the valuation $v_i$, for every subset $S \subseteq [m]$, as follows 
\begin{align*}
    v_i(S) \coloneqq \begin{cases}
    1 & \text{if }  P_j^i \subseteq S \text{ for any }j\in[n], \\
    0 & \text{otherwise}.
\end{cases}
\end{align*}
Note that $v_i$s are monotone set functions. In addition, the maximin share of every agent $i$ is equal to one, $\mu_i=1$. Hence, for any MMS multi-allocation $\calA=(A_1,\ldots,A_n)$ it must hold that $P^i_j \subseteq A_i$, for each agent $i \in [n]$ and some index $j \in [n]$. We can, in fact, restrict attention to those multi-allocations that satisfy $A_i=P_i^j$ for some $j \in [n]$, since the ones wherein $P^i_j \subsetneq A_i$ induce larger $\ell_\infty$ norms. Hence, given that partitions $\calP^i$ for the agents $i \in [n]$, we define the family $\calF$ of multi-allocations as 
\begin{align*}
    \calF \coloneqq \left\{ \calA=(A_1,\ldots,A_n) \mid \text{ for each } i\in[n] \text{ there exists } j\in[n] \text{ such that } A_i=P^i_j\right\}.
\end{align*}
Note that $|\mathcal{F}|=n^n$. Fix one multi-allocation $\calQ \in \calF$ obtained by setting index $j=1$ for all the agents $i$, i.e., $\calQ \coloneqq (P^1_1, P^2_1,\ldots, P^n_1)$. Write $\chi$ to denote the characteristic vector of $\calQ$. Now, consider the independent random draws that result in the  partitions $\calP^i=(P^i_1, \ldots P^i_n)$. Note that, for a fixed agent $i$ and a fixed good $g$, the probability that $g$ is contained in $i$'s bundle in $\calQ$ satisfies $\prob{g \in P^i_1} = 1/n$. Further, the independence of the draws (of $\calP^i$s) across the agents imply that the number of copies of $g$ assigned in $\calQ$ (i.e., $\chi_g$) is distributed as $\chi_g \sim \mathrm{Bin}(n, \frac{1}{n})$. For large $n$, this is approximated by the Poisson distribution $\mathrm{Poi}(n \  \frac{1}{n}) = \mathrm{Poi}(1)$. 
 
Define $\ell \coloneqq \frac{\log m}{\log \log m}$. Since each $\calP^i$ is drawn independently among all the $n$-partitions of $[m]$, the random variables $\chi_g$ across the goods $g \in [m]$ are independent. Therefore, using the fact that $\chi_g \sim \mathrm{Poi}(1)$ are independent and identically distributed for all $g \in [m]$, we obtain 
\begin{align}
    \prob{\ellinfty{\chi} \leq  \ell} = \left(\prob{\chi_g \leq \ell}\right)^m = \left(\sum_{k = 0}^{\ell} \frac{e^{-1}}{k!}\right)^m = \left(\frac{1}{e}\sum_{k = 0}^{\ell}\frac{1}{k!}\right)^m \label{eq:maxPoi}
\end{align}
Recall that $e = \sum_{k=0}^\ell \frac{1}{k!} + \sum_{k=\ell+1}^\infty \frac{1}{k!}$. 
 Hence, equation (\ref{eq:maxPoi}) simplifies to 
\begin{align*}
\prob{\ellinfty{\chi} \leq  \ell} = \left(\frac{1}{e} \left(e- \sum_{k=\ell+1}^\infty \frac{1}{k!} \right)\right)^m < \left( \frac{1}{e} \left(e - \frac{1}{(\ell+1)!} \right) \right)^m = \left(1 - \frac{1}{e(\ell+1)!}\right)^m.    \end{align*}
 
We can extend the last bound as follows
\begin{align}
 \prob{\ellinfty{\chi} \leq  \ell} < \left(1 - \frac{1}{e(\ell+1)!} \right)^{e(\ell+1)!\left(\frac{m}{e(\ell+1)!}\right)}  \leq \left(\frac{1}{e}\right)^{\frac{m}{e(\ell+1)!}} \label{ineq:endPoi}    
\end{align}
 
Since $\ell = \frac{\log m}{\log \log m}$, for a sufficiently large $m \in \mathbb{Z}_+$, we have $m \geq 2 \  n \log n \  e (\ell +1)! $. For such an $m$, inequality (\ref{ineq:endPoi}) reduces to 
\begin{align}
    \prob{\ellinfty{\chi} \leq  \ell} <  \left(\frac{1}{e}\right)^{2n \log n} = \frac{1}{n^{2n}} \label{ineq:finalPoi}.
\end{align}
Inequality (\ref{ineq:finalPoi}) holds for each fixed MMS multi-allocation $\calQ \in \calF$. Hence, a union bound over all the $n^n$ multi-allocations in $\calF$ establishes the theorem. In particular, define the event $E_\calQ: = \left\{ \ellinfty{\chi^\calQ} > \ell \right\}$ for each $\calQ \in\mathcal{F}$. We have 
\begin{align*}
\prob{\bigcap_{\calQ \in \mathcal{F}} E_\calQ} = 1- \prob{\bigcup_{\calQ \in\mathcal{F}} E_Q^c} \geq 1 - n^n \cdot \prob{E_Q^c} \underset{\text{ via (\ref{ineq:finalPoi}})}{>} 1-\frac{1}{n^n} > 0    
\end{align*}

Therefore, with a non-zero probability, the events $E_\calQ$ hold together for all possible MMS multi-allocations $\calQ$. This implies that there exists an instance with a sufficiently large $m$ and a certain choice of the partitions $\calP_i=\{P^i_1, \ldots, P^i_n)$ (along with valuations defined as above) such that for all MMS multi-allocations, there exists a good $g$ with at least $\ell = \frac{\log m}{\log \log m}$ assigned copies. This completes the proof of the lower bound. 
\end{proof}

\begin{remark}
We note that, for exact MMS, the $\ell_\infty$ lower bound given in \Cref{thm:goods_lowerbound} holds even under XOS valuations. In particular, as in the proof of \Cref{thm:goods_lowerbound}, we select independently for each agent $i \in [n]$, a partition $\calP^i = (P^i_1, P^i_2, \dots, P^i_n)$ uniformly at random from the set $\Pi_n([m])$ of all possible $n$-partitions of $[m]$. Then, for each agent $i \in [n]$ and subset $S \subseteq [m]$, we set the valuation
\begin{align*}
v_i(S) \coloneqq  \max_{1\leq j\leq n} \ \frac{1}{|P^i_j|} |S \cap P^i_j|. 
\end{align*}
Note that the valuations $v_i$ are point-wise maximizers of additive functions of the form $S \mapsto \frac{1}{|P^i_j|} |S \cap P^i_j|$. Hence, $v_i$s are XOS. Further, $v_i(S)=1$ if $P^i_j \subseteq S$, for some index $j \in [n]$, and $v_i(S) < 1$ otherwise. Therefore, the partition $\calP^i = (P^i_1, \dots, P^i_n)$ certifies that, for each agent $i \in [n]$, the maximin share $\mu_i = 1$. Moreover, for any MMS multi-allocation $\calA = (A_i, \dots, A_n)$ it must hold that $P^i_j \subseteq A_i$ for each agent $i \in [n]$ and some index $j \in [n]$. This is exactly the property we utilize for the proof of \Cref{thm:goods_lowerbound}, and the rest of the proof relies solely on this condition. Therefore, the lower bound holds in particular for XOS valuations. 
\end{remark}
\section{Fair Division of Chores}
\label{sec:chores}

This section addresses fair division instances wherein the items to be assigned are undesirable to---i.e., have costs for---the agents. Such an instance is specified by a triple $\langle[n],[m],\{c_i\}_{i=1}^n\rangle$ with $n$ agents and $m$ items, called \textit{chores}. Here, the function $c_i:2^{[m]}\to\mathbb{R}_{\geq 0}$ specifies the (nonnegative) \textit{cost} of each subset of chores to agent $i$. A cost function $c_i$ is said to be monotone if the inclusion of chores into any subset does not decrease its cost: $c_i(S) \leq c_i(T)$ for every pair of subsets $S \subseteq T \subseteq[m]$. Further, $c_i$ is said to be additive if for every subset $S \subseteq [m]$ of chores, $c_i(S) = \sum_{a \in S} c_i(\{a\})$. As a shorthand, we will write $c_i(a)$ to denote agent $i$'s cost for chore $a \in [m]$. Throughout, the paper assumes that the agents' costs are monotone and normalized: $c_i(\emptyset)=0$ for all agents $i$. 

We will additionally consider (in \Cref{subsec:ordered-costs}) instances with cost functions that are identically ordered. 

The maximin share of each agent $i \in [n]$ is the minimum cost incurred by $i$ if the agent were to partition the $m$ chores into $n$ subsets and is then assigned the one with the highest cost. Formally,
\begin{definition}[MMS for chores] 
Given any fair division instance $\langle [n], [m], \{c_i\}_{i=1}^n \rangle$ with chores, the {maximin share}, $\mu_i \in \mathbb{R}_+$, of each agent $i \in [n]$ is defined as 
\begin{align*}
\mu_i \coloneqq  \min_{(X_1,\dots, X_n) \in \Pi_n([m])} \ \ \max_{j\in[n]} c_i(X_{j}).
\end{align*}
Further, for each agent $i$, let $\calM^i=(M^i_1, M^i_2, \ldots, M^i_n) \in \Pi_n([m])$ denote an {MMS-inducing partition}:
\begin{align*}
\calM^i \in \argmin_{(X_1,\dots, X_n) \in \Pi_n([m])} \ \ \max_{j\in[n]} c_i(X_{j})
\end{align*}
\end{definition} 

As before, an allocation $\calB=(B_1,\ldots, B_n)$ is a partition of $[m]$ into $n$ pairwise disjoint subsets $B_1,\ldots, B_n \subseteq [m]$. Here, the subset of chores $B_i$ is assigned to agent $i$ and is referred to as $i$'s bundle. Also, a \textit{multi-allocation} is a tuple $\calA=(A_1,\ldots,A_n)$ of $n$ subsets of chores, wherein subset $A_i \subseteq [m]$ denotes the bundle assigned to agent $i$. In contrast to allocations, in a multi-allocation, we do not require that the assigned bundles $A_i$ are pairwise disjoint and that they partition $[m]$. Hence, in a multi-allocation, a single chore may be present in multiple bundles or even in none.

\paragraph{Fair Multi-Allocations.} A multi-allocation $\calA=(A_1,\dots,A_n)$ is said to be an \emph{MMS multi-allocation} if each agent receives a bundle of cost at most its maximin share, $c_i(A_i)\leq \mu_i$ for all agents $i \in [n]$.

Agents with monotone costs prefer to have fewer chores in their bundle and, hence, additional copies of any chore can be disposed of without impacting the MMS guarantee. Specifically, given any multi-allocation $\calA$ in which one chore $a$ is assigned to agents $i \neq i'$, one can consider multi-allocation $\calA'$ in which $a$ is assigned to only one of them, say $i$. If $\calA$ is MMS, then so is $\calA'$. That is, one can assume, without loss of generality, that the considered MMS multi-allocations $\calA$ satisfy $\ellinfty{\chi^\calA} \leq 1$, i.e., $\chi^\calA_a$ is always either $0$ or $1$ for each chore $a \in [m]$.\footnote{This makes such a multi-allocation in fact a \textit{partial allocation}.} Motivated by these considerations,  we will focus on quantifying the number of distinct chores that can be assigned among the agents while ensuring fairness. That is, in contrast to the goods setting where we obtained upper bounds for $\ell_1$ and $\ell_\infty$, here, we will focus on establishing upper bounds on the number of zero entries in the vector $\chi^\calA$. We will, throughout, write $\ellzed{\chi^\calA}$ to denote this quantity, $\ellzed{\chi^\calA} \coloneqq \{ a \in [m] \mid \chi^\calA_a = 0 \}$.\footnote{Note that this is not a norm.} Observe that $\ellzed{\chi^\calA}$ captures the number of chores that are unallocated under $\calA$; ideally, this quantity should be as small as possible.

\subsection{Monotone Costs}
For instances with monotone costs, the following theorem shows that we can achieve MMS fairness while at most $m/e$ chores remain unassigned.

\begin{theorem}
\label{thm:monotone_chores}
    Every fair division instance $\langle [n], [m], \{c_i\}_{i=1}^n \rangle$ with chores and monotone costs admits an MMS multi-allocation $\calA=(A_1,\dots,A_n)$ in which at most $\frac{m}{e}$ chores remain unassigned. That is, the characteristic vector $\chi^\calA$ of $\calA$ satisfies $\ellzed{\chi^\calA}\leq\frac{m}{e}$. 
\end{theorem}

\begin{proof}
    Recall that $\calM^i=(M^i_1,\ldots,M^i_n)$ is an MMS-inducing partition for each agent $i \in [n]$. We select, independently for each $i$, a bundle $R_i$ uniformly at random from among $\{M^i_1, M^i_2, \dots,M^i_n\}$, and consider the random multi-allocation $\calR=(R_1,R_2,\dots,R_n)$. For any fixed chore $a\in[m]$, the events $\{a\notin R_i\}$ are independent across $i$ with $\prob{a\notin R_i}=1-\frac{1}{n}$. Therefore, under $\calR$, the probability that $a$ remains unallocated  $\prob{a\notin\cup_{i=1}^n R_i}=\left(1-\frac{1}{n}\right)^n \simeq \frac{1}{e}$. Hence, the total number of unallocated chores, in expectation, equals 
    \begin{align*}
    \expval{\ellzed{\chi^\calR}}
    &= \sum_{a=1}^m \prob{a\notin\cup_{i=1}^n R_i} = \frac{m}{e}.
    \end{align*} 
    Therefore, there exists an MMS multi-allocation $\calA$ that satisfies $\ellzed{\chi^\calA}\leq \frac{m}{e}$. The theorem stands proved. 
\end{proof}

\subsection{Identically Ordered Costs}
\label{subsec:ordered-costs}
This section establishes upper bounds on supply adjustments under cost functions are identically ordered; the setup here is analogous to the one given in \Cref{subsec:additive-ordered} for goods. Formally, 

\begin{definition} \label{def:identically-ordered-costs}
In a fair division instance $\langle [n], [m], \{c_i\}_{i=1}^n \rangle$, the cost functions are said to be \textit{identically ordered} if there exists an indexing $a_1, a_2, \ldots, a_m$ over the chores such that, for any pair of chores $a_s, a_t$, with index $s < t$, we have $c_i(S + a_s) \geq c_i(S + a_t)$ for all agents $i \in [n]$ and each subset $S$ that does not contain $a_s$ and $a_t$.  
\end{definition}

We show that for any fair division instance with identically ordered costs, there exists a multi-allocation $\calA$ such that the number of unassigned chores in $\calA$ is $\widetilde{O} \left(\frac{m}{n^{1/4}} \right)$.


\begin{theorem}
\label{thm:additive_chores}
    Every fair division instance $\langle [n], [m], \{c_i\}_{i=1}^n \rangle$ with chores and identically ordered costs admits an MMS multi-allocation $\calA=(A_1,\ldots A_n)$ wherein the number of unassigned chores is at most $\frac{3m \left( \log m \right)^{3/2}}{n^{1/4}}$, i.e., 
    $$\ellzed{\chi^\calA} \leq \frac{3m \left( \log m \right)^{3/2}}{n^{1/4}}.$$
\end{theorem}

\paragraph{Pre-Processing and Random Sampling.} For each agent $i$, consider an MMS-inducing partition $\calM^i=(M^i_1,\ldots,M^i_n)$. Rather than sampling an MMS multi-allocation directly, we first execute a pre-processing step, as detailed in Algorithm \ref{alg:chores-additive}. The algorithm is described in terms of a parameter $\alpha>0$, which will be fixed later.

\begin{algorithm}
\caption{Pre-processing for Identically Ordered Costs}\label{alg:chores-additive}
\begin{algorithmic}[1]
    \REQUIRE An instance $\calI = \langle [n], [m], \{c_i\}_{i=1}^n \rangle$ along with MMS-inducing partitions $\calM^i=(M^i_1,\ldots,M^i_n)$, for the agents $i$, and parameter $\alpha>0$.
    \ENSURE A multi-allocation $\calB=(B_1,B_2,\dots,B_n)$ and sub-instance $\calI'$
    \STATE Initialize $N=[n]$, $U=[m]$ and $B_i=\emptyset$, for each agent $i\in[n]$.
    \WHILE{there exists agent $i\in N$ and index $j\in[n]$ such that $|M^i_j \cap U|\geq \alpha \ \frac{|U|}{|N|}$}
        \STATE Set $B_i=M^i_j\cap U$. 
        \STATE Update $N \leftarrow N\setminus \{i\}$ and $U \leftarrow U\setminus M^i_j$.
    \ENDWHILE
    \RETURN $\calB$ and sub-instance $\calI'=\langle N, U, \{ c_i \}_{i \in N} \rangle$.
\end{algorithmic}
\end{algorithm}

\Cref{alg:chores-additive} starts with $U = [m]$ as the set of unassigned chores and $N = [n]$ as the set of agents under consideration. The algorithm iterates as long as it is possible to assign to some agent $i \in N$ an MMS-inducing bundle, from $\calM^i$, of sufficiently large cardinality. Specifically, if for some agent $i \in N$ there exists a bundle $M^i_j$ such that $|M^i_j \cap U| \geq \alpha \ \frac{|U|}{|N|}$, then agent $i$ receives the bundle $B_i = M^i_j \cap U$. After this assignment, the algorithm removes both the agent $i$ and the assigned bundle from the instance. Since $c_i(B_i)=c_i(M^i_j\cap U)\leq c_i(M^i_j)\leq \mu_i$, each removed agent receives a bundle of cost at most its MMS.

After the algorithm terminates, we are left with the reduced sub-instance $\calI' = \langle N, U, \{c_i\}_{i\in N}\rangle$, where $N$ and $U$ are the sets of remaining agents (those yet to receive a bundle) and unallocated chores, respectively. A key property of $\calI'$ is that, for each remaining agent $i \in N$, the $n$ MMS-satisfying subsets $M^i_1 \cap U, M^i_2 \cap U, \ldots, M^i_n \cap U$ have size $|M^i_j \cap U| < \alpha \  \frac{|U|}{|N|}$. We will utilize this property and assign chores in $\calI'$ next. 

Note that, in $\calI'$, if $U = \emptyset$ (i.e., no chores remain unassigned at the end of the algorithm), then the upper bound stated in Theorem \ref{thm:additive_chores} holds. In fact, in such a case, for the underlying instance $\calI$, we can identify an MMS multi-allocation $\calA$ with $\ellzed{\chi^\calA}=0$, by setting $A_j = \emptyset$ for all $j \in N$ and $A_i = B_i$ for all $i \in [n] \setminus N$. 

Hence, we will assume that $|U|\geq 1$. Note that all the chores in $[m]\setminus U$ are assigned to some agent in the multi-allocation, $\calB$, returned by the algorithm. Hence, for any $k \in \mathbb{Z}_+$, the existence of an MMS multi-allocation $\calA'$ in $\calI'$ with $\ellzed{\chi^{\calA'}} \leq k$ (i.e., under $\calA'$ at most $k$ chores from $U$ remain unassigned) implies the existence of an MMS multi-allocation $\calA$ in $\calI$ with $\ellzed{\chi^\calA} \leq k$. In particular, we can obtain $\calA$ from $\calA'=(A'_j)_{j \in N}$ by setting $A_i = B_i$ for all $i \in [n] \setminus N$ and $A_j = A'_j$ for all $j \in N$. 

Towards finding such an MMS multi-allocation $\calA'$ in $\calI'$, with the upper bound $k$ as stated in Theorem \ref{thm:additive_chores}, we perform the following random sampling. Independently for each agent $i\in N$, we select a bundle $R_i$ uniformly at random from the set $\{M^i_1\cap U,\dots,M^i_n\cap U\}$, and consider the random multi-allocation $\calR=(R_i)_{i\in N}$ in $\calI'$. As mentioned previously, $ |M^i_j \cap U| < \alpha \  \frac{|U|}{|N|}$ for each agent $i \in N$ and index $j \in [n]$ and, hence, for any random draw, we have   
\begin{align}
    |R_i| < \alpha \ \frac{|U|}{|N|} \quad \text{for each $i \in N$} \label{ineq:sizeR}
\end{align}   

The following lemma provides a useful lower bound on the expected value of the components of $\chi^\calR$.

\begin{lemma}\label{lem:rem-agents}
    In the reduced instance $\calI'=\langle N,U,\{c_i\}_{i\in N}\rangle$ and under the above-mentioned random draws of MMS allocations $\calR=(R_i)_{i\in N}$, we have  
    \begin{align*}
        \expval{\chi^\calR_a}\geq 1-\frac{\log m}{\alpha} \quad \text{for each chore $a\in U$.}
    \end{align*}
\end{lemma}
\begin{proof}
    For the given instance $\calI = \langle [n], [m], \{c_i\}_{i=1}^n \rangle$, write $T$ to denote the number of iterations of the while-loop in \Cref{alg:chores-additive}. Note that the number of agents, $|N|$, remaining in the reduced instance $\calI'$ satisfies $|N| = n - T$. Also, let $U_t$ and $N_t$ denote the sets of remaining chores and agents, respectively, after $t \in \{0,\dots,T\}$ iterations of the while-loop. Note that, during each iteration $t+1$, at least $\alpha\ \frac{|U_t|}{|N_t|}$ chores are assigned to the selected agent. Hence, for each $t\in\{0,\dots,T\}$: 
    \begin{align*}
    |U_{t+1}| \leq|U_t|-\alpha\ \frac{|U_t|}{|N_t|}=|U_t|\ \left(1-\frac{\alpha}{|N_t|}\right) 
     \leq|U_t|\ \left(1-\frac{\alpha}{n}\right).
    \end{align*} 
    Therefore, after $T$ iterations 
    \begin{align}|U_T| \leq m \  \left( 1 - \frac{\alpha}{n} \right)^T \label{ineq:size-of-u} 
    \end{align}
    Since $|U_T|\geq 1$, equation (\ref{ineq:size-of-u}) yields $1-\frac{\alpha}{n}\geq \left(\frac{1}{m}\right)^\frac{1}{T}$. Exponentiating both sides of the previous inequality by  $\frac{n}{\alpha}$ gives us $\frac{1}{e}\simeq (1 - \frac{\alpha}{n})^{\frac{n}{\alpha}} \geq \left(\frac{1}{m}\right)^\frac{n}{\alpha T}$. Simplifying, we obtain $T\leq  \frac{n}{\alpha} \ \log m$ and, hence, 
    \begin{align}
        |N| = |N_T| = n- T \geq n  \left(1-\frac{\log m}{\alpha}\right) \label{ineq:size-of-N}
    \end{align}
    Finally, note that for each chore $a\in U$ 
    \begin{align*}
        \expval{\chi^\calR_a}= \sum_{i\in N} \prob{a\in R_i} = \frac{|N|}{n} \underset{ \text{via (\ref{ineq:size-of-N})}}{\geq} 1-\frac{\log m}{\alpha}.
    \end{align*}
This completes the proof of the lemma. 
\end{proof}

Since the underlying instance $\calI = \langle [n], [m], \{c_i\}_{i=1}^n \rangle$ consists of identically ordered costs, this property continues to hold for the reduced instance $\calI'=\langle N,U,\{c_i\}_{i\in N}\rangle$ as well. Specifically, we will assume, without loss of generality, that the set $U$ consists of chores $a_1,a_2, \ldots, a_{|U|}$, indexed such that for any $a_s$ and $a_t$, with $s<t$, we have $c_i(S + a_s) \geq c_i(S+a_t)$, for all agents $i \in N$ and all subsets $S \subset U$ that do not contain $a_s$ and $a_t$. We now reason about the number of chores assigned, with copies, under the random multi-allocation $\calR$ and in any given sequence of chores $S=\{a_r,a_{r+1},\ldots,a_{s-1}, a_s\}$. We denote by $\mathcal{S}$ the collection of all chore sequences in $U$, i.e., 
\begin{align*}
\mathcal{S}=\left\{ \{a_r,a_{r+1},\dots,a_s\}\mid 1\leq r\leq s\leq |U| \right\}.
\end{align*}

\begin{lemma}\label{lem:ordered-sequences}
In the reduced instance $\calI'=\langle N,U,\{c_i\}_{i\in N}\rangle$ and under the above-mentioned random draws of MMS allocations $\calR=(R_i)_{i\in N}$, the following bound holds for each chore sequence $S \in \mathcal{S}$: 
\begin{align*}
    \prob{\sum_{a\in S} \chi^\calR_a < |S|-\Delta}\leq \frac{1}{m^4},
    \end{align*}
    where 
    \begin{align*}
        \Delta \coloneqq |S|\ \frac{\log m}{\alpha} +2\alpha\ \frac{|U| \sqrt{\log m}}{\sqrt{|N|}}.
    \end{align*}
\end{lemma}

\begin{proof}
    \Cref{lem:rem-agents} gives us 
    \begin{align}
        \expval{\sum_{a\in S}\chi^\calR_a}\geq|S|  \left(1-\frac{\log m}{\alpha}\right) \label{eqn:exp-val-chore}
    \end{align} 
Note that $\sum_{a\in S}\chi^\calR_a=\sum_{i\in N}|R_i\cap S|$, where the random variables $|R_i\cap S|$ are independent across $i \in N$. Also, equation (\ref{ineq:sizeR}) provides a range for these random variables: $0\leq |R_i\cap S| <  \alpha\ \frac{|U|}{|N|}$, for each $i \in N$.
    
We now apply Hoeffding's inequality (lower tail) over the sum of independent random variables $|R_i\cap S|$, with parameter $\delta=2\alpha\ \frac{|U|}{\sqrt{|N|}}\ \sqrt{\log m}$:  
\begin{align}
\prob{\sum_{a\in S}\chi^\calR_a \leq|S|\ \left(1-\frac{\log m}{\alpha}\right)-\delta}
    &\leq \prob{\sum_{a\in S}\chi^\calR_a \leq \expval{\sum_{a\in S}\chi^\calR_a}-\delta} \tag{via (\ref{eqn:exp-val-chore})}\\
    & = \prob{ \sum_{i\in N}|R_i\cap S| \leq \expval{\sum_{i\in N}|R_i\cap S|}-\delta} \nonumber \\ 
    & \leq \exp \left( \frac{- \delta^2}{|N| \alpha^2 \frac{|U|^2}{|N|^2} } \right) \nonumber \\ 
    &\leq \exp(-4\log m) \tag{via value of $\delta$} \\
    &=\frac{1}{m^4} \label{ineq:tail}.
\end{align} 
    Inequality (\ref{ineq:tail}) corresponds to the bound stated in the lemma statement, since
    \begin{align*}
    |S|\ \left(1-\frac{\log m}{\alpha}\right)-\delta &= |S| - |S| \frac{\log m}{\alpha} - 2\alpha \frac{|U|}{\sqrt{|N|}}\ \sqrt{\log m} 
    = S-\Delta.
    \end{align*}
    Recall that $\Delta=|S|\frac{\log m}{\alpha} +2\alpha \frac{|U| \sqrt{\log m}}{\sqrt{|N|}}$. The lemma stands proved. 
\end{proof}

The following lemma shows that covering guarantees with respect to chore sequences, such as the ones obtained in Lemma \ref{lem:ordered-sequences}, translate to upper bounds on the number of unassigned chores. 

\begin{lemma}
\label{lem:shifting_chores}
In fair division instance $\calI'=\langle N,U,\{c_i\}_{i\in N}\rangle$ and for any parameter $\Delta\in \mathbb{Z}_{\geq 0}$, let $\calQ=(Q_i)_{i \in N}$ be a multi-allocation whose characteristic vector $\chi^\calQ$ satisfies $\sum_{a\in S} \chi^\calQ_a\geq |S|-\Delta$ for every chore sequence $S\in\mathcal{S}$. Then, $\calI'$ admits a multi-allocation $\calA'=(A'_i)_{i \in N}$ with the properties that 
\begin{itemize}
        \item[(i)] $c_i(A'_i)\leq c_i(Q_i)$ for all the agents $i$.
        \item[(ii)] $\ellzed{\chi^{\calA'}} \leq \Delta$.
    \end{itemize}
\end{lemma}

\begin{algorithm}
\caption{Copy Redistribution for Identically Ordered Costs}\label{alg:chores-shifting-2}

\begin{algorithmic}[1]
    \REQUIRE A multi-allocation $\calQ$ satisfying the condition of \Cref{lem:shifting_chores}.
    \ENSURE A multi-allocation $\calA'$.
    \STATE Initialize $\calA'=\calQ$.
    \WHILE{there exist chores $b$ and $b'$ with the following properties: (i) index of $b$ is less than that of $b'$ (i.e., $b$ has higher marginal costs than $b'$), (ii) $\chi^{\calA'}_{b}\geq 2$, and (iii) $\chi^{\calA'}_{b'}=0$}
    \STATE Let $a_s$ be the lowest-indexed chore with $\chi^{\calA'}_{a_s}\geq 2$ and let $a_t$ be the lowest-indexed chore such that $t > s$ and $\chi^{\calA'}_{a_t} = 0$. \label{line:lowchore}
    \STATE Select an agent $i$ with $a_s\in A'_i$ and update $A'_i \leftarrow (A'_i \setminus \{a_s \} ) \cup \{a_t\}$.
    \ENDWHILE
    \RETURN $\calA'$
\end{algorithmic}
\end{algorithm}

\begin{proof}
    We will show that---given a multi-allocation $\calQ$ that satisfies the stated guarantees with respect to chore sequences---\Cref{alg:chores-shifting-2} returns the desired multi-allocation $\calA'$. \Cref{alg:chores-shifting-2}, in each iteration of its while-loop, replaces a chore $a_s$ in the selected agent $i$'s bundle, $A'_i$, by another chore $a_t$ of higher index $t > s$. Since $a_s$ was previously allocated to at least two agents and $a_t$ was unallocated, the number of unallocated chores in the maintained multi-allocation $\calA'$ decreases by one in each iteration. That is, the algorithm terminates in polynomial time. Moreover, given that $a_t$ has a higher marginal cost than $a_s$ with respect to $A'_i \setminus \{a_s\}$, the agents' allocated costs $c_i(A'_i)$ are non-increasing throughout the execution of the algorithm. Hence, for the returned multi-allocation $\calA'$, in particular, we have $c_i(A'_i) \leq c_i(Q_i)$, for all the agents $i$. This observation establishes property (i) as stated in the lemma.

    We next establish property (ii) by proving that upon termination of the algorithm, the returned multi-allocation $\calA'$ satisfies $\ellzed{\chi^{\calA'}} \leq \Delta$. Assume, towards a contradiction, that $\ellzed{\chi^{\calA'}}\geq \Delta+1$. Write $a_r$ to denote the chore of highest index $r$ with the property $\chi^{\calA'}_{a_r}=0$. Also, let $S$ denote the sequence of chores in $U$ with index at most $r$. Since $a_r$ is the highest index chore that is not assigned ($\chi^{\calA'}_{a_r}=0$), we have $\chi^{\calA'}_{b'} \geq 1$ for all chores $b'$ with index higher than $r$. Hence, all the unassigned chores are in fact contained in $S$; recall that there are at least $(\Delta +1)$ unassigned chores under $\calA'$. Further, the fact that the while-loop of the algorithm terminated with multi-allocation $\calA'$ and $\chi^{\calA'}_{a_r}=0$ ensures that $\chi^{\calA'}_b \leq 1$ for all chores $b \in S$. These observations imply that 
    \begin{align}
        \sum_{b \in S} \chi^{\calA'}_b \leq |S| - \Delta -1 \label{ineq:prefix-delta}
    \end{align}
    Further, note that, since $\chi^{\calA'}_{a_r} = 0$, during any iteration of the while-loop, the algorithm would not have replaced a copy of a chore $a_s \in S$ with a copy of a chore $a_t \notin S$. That is, it could not have been the case that, during any iteration, the considered chores $a_s$ and $a_t$ satisfy $s < r < t$, since such a relation between the indices contradicts the selection criterion of $a_t$ as the lowest-index chore with $\chi^{\calA'}_{a_t} = 0$. Therefore, throughout the algorithm's execution, for any considered chore $a_s \in S$, the corresponding chore $a_t$ must have also been contained in $S$. Considering this property along with the decrements and increments of $\chi^{\calA'}_{a_s}$ and $\chi^{\calA'}_{a_t}$, respectively, we obtain $\sum_{b \in S} \chi^\calQ_b = \sum_{b \in S} \chi^{\calA'}_b$. 

    The last equality and equation (\ref{ineq:prefix-delta}) give us $\sum_{b \in S} \chi^\calQ_b \leq |S| - \Delta -1$. This bound, however, contradicts the condition provided for $\calQ$ in the lemma statement. Hence, by way of contradiction, we obtain property (ii), $\ellzed{\chi^{\calA'}} \leq \Delta$. The lemma stands proved. 
\end{proof}

\medskip

    \begin{proof}[Proof of \Cref{thm:additive_chores}]

    Let $\calI=\langle[n],[m],\{c_i\}_{i=1}^n\rangle$ be the original fair division instance with chores and additive ordered costs. We first execute the pre-processing via \Cref{alg:chores-additive} (with parameter $\alpha \geq 1$, to be fixed later) and obtain a multi-allocation $\calB=(B_1,\ldots,B_n)$ along with sub-instance $\calI'=\langle N,U,\{c_i\}_{i\in N}\rangle$. Recall that $N$ and $U$ are the sets of remaining agents and unallocated chores, respectively. Also, all the chores in $[m] \setminus U$ are assigned among the agents in the set $[n] \setminus N$ with MMS guarantees; in particular, $\cup_{i \in [n] \setminus N} B_i = [m] \setminus U$ and $c_i(B_i) \leq \mu_i$ for each agent $i \in [n] \setminus N$.  
    
    We then focus on the sub-instance $\calI'=\langle N,U,\{c_i\}_{i\in N}\rangle$ with the goal of assigning chores from $U$ while ensuring MMS among the agents in $N$. Towards this and as mentioned previously, we sample an MMS multi-allocation $\calR=(R_i)_{i\in N}$ by drawing $R_i$ uniformly at random from $\{M^i_j\cap U\}_{j=1}^n$ for each $i\in N$. \Cref{lem:ordered-sequences} ensures that, for each chore sequence $S\in \mathcal{S}$, we have  
    \begin{align}
        \prob{\sum_{a\in S} \chi^\calR_a < |S|-\Delta}\leq \frac{1}{m^4}  \qquad \text{where } \Delta =|S|\ \frac{\log m}{\alpha} +2\alpha\ \frac{|U| \sqrt{\log m}}{\sqrt{|N|}} \label{eq:each-sequence}
    \end{align} 
    We fix $\alpha \coloneqq  n^{1/4} \log m$. This choice and inequality (\ref{ineq:size-of-N}) gives us $|N| \geq n \left(1 - \frac{\log m}{\alpha} \right) = n - n^{3/4} = n(1 - o(1))$. Further, 
     \begin{align}
        \Delta & =|S|\ \frac{\log m}{\alpha} \ + \ 2\alpha\ \frac{|U| \sqrt{\log m}}{\sqrt{|N|}} \nonumber \\ 
        & \leq \frac{m\log m}{\alpha} \ + \ \frac{2\alpha m\sqrt{\log m}}{\sqrt{|N|}} \tag{since $|S| \leq |U| \leq m$} \\
        & = \frac{m}{n^{1/4}}  \ + \ \frac{2  m n^{1/4} \left( \log m \right)^{3/2}}{ \sqrt{|N|} } \tag{$\alpha = n^{1/4} \log m$ and $|N| \geq n(1- o(1))$} \\
        & \leq \frac{3  m \left(\log  m \right)^{3/2}}{n^{1/4}} \label{eq:eq_alpha}
    \end{align}

    Write $\Delta^*$ to denote the upper bound on the deviation obtained in equation (\ref{eq:eq_alpha}), $\Delta^* \coloneqq \frac{3m \left(\log  m\right)^{3/2}}{n^{1/4}}$. Given that $|\mathcal{S}| \leq |U|^2\leq m^2$, applying union bound over all the chore sequences in $\mathcal{S}$ gives us  
    \begin{align*}
        \prob{ \sum_{a\in S} \chi^\calR_a \geq |S|-\Delta^*  \ \text{ for each } S\in\mathcal{S}}
    \geq 1 - \sum_{S \in \mathcal{S}} \prob{\sum_{a \in S} \chi^\calR_a < |S|-\Delta^*} \underset{\text{via (\ref{eq:each-sequence})}}{\geq} 
    1-\frac{m^2}{m^4}>0.
    \end{align*}
    Since this probability is positive, there exists an MMS multi-allocation $\calQ=(Q_i)_{i\in N}$ that satisfies $\sum_{a\in S} \chi^{\calQ}_a \geq |S|-\Delta^*$ for every chore sequence $S \in \mathcal{S}$. This is the condition required to invoke \Cref{lem:shifting_chores}, which, in turn, implies the existence of an MMS multi-allocation $\calA'=(A'_i)_{i\in N}$ (in $\calI'$) with $\ellzed{\chi^{\calA'}}\leq \Delta^*$. The final multi-allocation $\calA=(A_1,\ldots,A_n)$ for the underlying instance $\calI$ is obtained by combining $\calA'$ and $\calB$: For each agent $i$, set 
    \begin{align*}
        A_i=\begin{cases}
        A'_i & \text{ if }i\in N,\\
        B_i &\text{ otherwise, if $i \in [n] \setminus N$}.
    \end{cases}
    \end{align*}
    By construction, for each agent $i \in [n]$, we have $c_i(A_i) \leq \mu_i$. Hence, $\calA$ is an MMS multi-allocation. Also, since all the chores in $[m] \setminus U$ are assigned under $\calB$, we have that $\ellzed{\chi^\calA} = \ellzed{\chi^{\calA'}} \leq \Delta^*$. Overall, we get that $\calA$ satisfies the properties stated in the theorem. This completes the proof.  
    \end{proof}

    \subsection{Lower Bound for Chores}

    We now prove that there exists fair division instances, with $m$ chores and monotone costs, such that in every MMS multi-allocation, at least $(1- o(1))\nicefrac{m}{e}$ chores remain unassigned. This lower bound shows that the positive result obtained in Theorem \ref{thm:monotone_chores} is essentially tight. 

    \begin{theorem}
    \label{thm:chores_lowerbound}

 For any $\delta>0$, there exists a fair division instance $\calI = \langle [n], [m], \{c_i\}_{i=1}^n \rangle$ with monotone costs such that under every MMS multi-allocation $\calA$ in $\calI$ at least $(1-\delta)\frac{m}{e}$ chores remain unassigned, i.e., for every MMS multi-allocation $\calA$ in $\calI$ we have $\ellzed{\chi^\calA}\geq (1-\delta)\frac{m}{e}$. 
\end{theorem}

\begin{proof}
    We will construct an instance with $n$ agents and $m$ chores with $m> \frac{2e}{\delta^2}n\log n$. For each agent $i$, we independently draw an $n$-partition $\calP^i=(P^i_1,P^i_2,\dots,P^i_n)$ of the set $[m]$ of chores uniformly at random from the set of all such partitions, i.e., $\calP^i \in_R \Pi_n([m])$. Given such a partition, we set the cost of each subset $S\subseteq [m]$ as 
    \begin{align*}
        c_i(S) =\begin{cases}
        0  & \text{if } S\subseteq P^i_j \text{ for any $j \in [n]$,}\\
        1  & \text{otherwise}.
    \end{cases}
    \end{align*}

     Note that each $c_i$ is a monotone set function. Further, for each agent $i$ the MMS value $\mu_i=0$, which implies that for any MMS multi-allocation $\calA=(A_1,\ldots,A_n)$ we have  $A_i\subseteq P^i_j$ for each $i\in [n]$ and some index $j\in[n]$. We may, in fact, restrict ourselves to those multi-allocations that satisfy $A_i=P^i_j$, for some $j$, since those that do not will only leave more chores unallocated. Hence, we define the family $\calF$ of MMS multi-allocations as 
     \begin{align*}
         \calF=\left\{ (A_1,A_2,\dots,A_n) \mid \text{ for each } i\in[n], \text{ there exists }j\in[n]\text{ such that } A_i=P^i_j\right\}.
     \end{align*} 
     Note that $|\mathcal{F}|=n^n$. Fix the multi-allocation $\calQ\in\calF$ obtained by setting index $j=1$ for all agents $i$, i.e. $\calQ=(P^1_1,P^2_1,\dots,P^n_1)$. Considering the random draws that result in the partitions $\calP^i= (P^i_1,\dots,P^i_n)$, for each fixed agent $i$ and chore $a$, we have $\prob{a\in P^i_i}=1/n$. Further, the fact that these draws are independent across the agents gives us $\prob{a\notin\cup_{i=1}^n P^i_1} = \left(1-\frac{1}{n}\right)^n$. Hence, 
     \begin{align*}
     \expval{\ellzed{\chi^\calQ}} =\sum_{a=1}^m \prob{a\notin\cup_{i=1}^n P^i_1} = \sum_{a=1}^m \left(1-\frac{1}{n}\right)^n \simeq\frac{m}{e}.
     \end{align*} 
     We next utilize Chernoff bound (lower tail), with $\delta>0$, to obtain 
     \begin{align}
     \prob{\ellzed{\chi^{\calQ}} < (1-\delta)\frac{m}{e}} \leq \exp\left(-\frac{\delta^2m}{2e}\right)< \frac{1}{n^n} \label{eq:each-A}
     \end{align} 
     The last inequality follows since $m> \frac{2e}{\delta^2}n\log n$. Note that the bound (\ref{eq:each-A}) holds for each MMS multi-allocation $\calQ \in \calF$. Let $E_\calQ$ denote the event that $\ellzed{\chi^{\calQ}}\geq (1-\delta)\frac{m}{e}$ and note that $\prob{E^c_\calQ} \leq \frac{1}{n^n}$. Now, applying union bound over all $\calQ \in \calF$ gives us  
     \begin{align*}
         \prob{\bigcap_{\calQ \in \calF} E_\calQ} = 1-\prob{\bigcup_{\calQ\in\calF} E_\calQ^c}\geq 1-n^n\  \prob{E_\calQ^c}>1-\frac{n^n}{n^n}=0.
     \end{align*}
     Therefore, with non-zero probability, the events $E_\calQ$ hold together for all MMS multi-allocations $\calQ\in \calF$. This implies that there exists an instance $\calI$, determined by the choice of partitions $\calP^i=(P^i_1,\dots,P^i_n)$, such that every MMS multi-allocation in $\calI$ leaves at least $(1-\delta)\frac{m}{e}$ chores unallocated. This completes the proof.

\end{proof}

\subsection{Additive Costs}

When agents' costs are additive, the following theorem shows that MMS fairness can be achieved by leaving at most $\frac{2}{11}m + n$ chores unassigned.

\begin{theorem}
    \label{thm:chores-additive}
        Every fair division instance with chores and additive costs admits an MMS multi-allocation $\calA = (A_1, \dots, A_n)$ in which at most $\frac{2}{11}m + n$ chores remain unassigned. That is, the characteristic vector $\chi^\calA$ of $\calA$ satisfies $\ellzed{\chi^\calA} \leq \frac{2}{11}m + n$.
    \end{theorem}
    
    \begin{proof}
        The work of Huang and Lu \cite{10.1145/3465456.3467555} shows that fair division instances with chores and additive costs admit an exact allocation that assigns to each agent a bundle of cost no more than $\frac{11}{9}$ times its MMS. Let $\calA' = (A'_1, \dots, A'_n)$ be such an allocation for the given instance. Specifically, we have \begin{equation} c_i(A'_i) \leq \frac{11}{9} \mu_i \label{eq:chores-approx}
        \end{equation} for each agent $i \in [n]$.
    
        Consider the bundle $A'_i$ received by some fixed agent $i$ under $\calA'$, and assume that $|A'_i| = k_i$. Write $A'_i = \{a_1, \dots, a_{k_i}\}$ to denote the indexing of the elements of $A'_i$ that satisfies $c_i(a_1) \geq \dots \geq c_i(a_{k_i})$. Define $p_i \coloneqq \left\lceil \frac{2}{11} k_i \right\rceil$ and $U_i \coloneqq \{a_1, \dots, a_{p_i}\}.$ Since $c_i(a_1) \geq \dots \geq c_i(a_{p_i}) \geq \dots \geq c_i(a_{k_i})$, an averaging argument gives $\frac{1}{p_i} c_i(U_i) \geq \frac{1}{k_i} c_i(A'_i)$ and, hence, 
        \begin{equation} c_i(U_i) \geq \frac{p_i}{k_i} c_i(A'_i) \geq \frac{2}{11} c_i(A'_i). \label{eq:unassigned}
        \end{equation}
    
        Now, define the multi-allocation $\calA = (A_i, \dots, A_n)$ by setting $A_i = A'_i \setminus U_i$ for each agent $i$. Then, $$v_i(A_i) 
        = v_i(A'_i) - v_i(U_i) 
        \underset{\text{via (\ref{eq:unassigned})}}{\leq} \frac{9}{11} c_i(A'_i) 
        \underset{\text{via (\ref{eq:chores-approx})}}{\leq} \mu_i.$$ That is, $\calA$ is an MMS multi-allocation. Moreover, the set of unassigned chores in $\calA$ is exactly $U \coloneqq \cup_{i=1}^n U_i$. Hence, 
        \begin{align*} \ellzed{\chi^\calA} 
            = |U| 
            = \sum_{i=1}^n |U_i|
            = \sum_{i=1}^n p_i 
            = \sum_{i=1}^n \left\lceil \frac{2}{11} k_i \right\rceil 
            \leq \sum_{i=1}^n \left( \frac{2}{11} k_i + 1 \right) 
            = \frac{2}{11}m + n. 
        \end{align*}

        The theorem stands proved.
    \end{proof}
 
\section{Conclusion and Future Work}
This work establishes that exact MMS fairness can be guaranteed in settings that permit post facto adjustments to the supply of indivisible items. We obtain tight bounds for general monotone valuations and monotone costs. Our specialized bounds for identically ordered valuations and costs bring forth novel applications of concentration inequalities in the fair division context. The techniques also highlight a distinction between goods and chores in the MMS context -- for the goods case, the problem in the current context corresponds to a vector packing problem. By contrast, in the chores case we have a vector covering problem.   

The framework of adjusted supply provides a meaningful route for bypassing barriers in discrete fair division. In particular, it would be interesting to develop truthful mechanisms for MMS that leverage bounded supply adjustment. Another relevant direction for future work is to consider the framework for achieving envy-freeness up to any good (EFX). Such an approach would complement prior works on EFX with charity (see, e.g., \cite{chaudhury2021little}), which obtain existential guarantees for EFX with disposal, rather than duplication, of goods.

\bibliographystyle{alpha}
\bibliography{refs}

\appendix
\section{Maximin Shares with Arbitrary Entitlements}
\label{section:entitlements}
We have thus far studied fair division among agents who have equal entitlements over the items. The current section complements this treatment and addresses agents with different entitlements. In particular, we consider settings in which each agent $i \in [n]$ is endowed with an entitlement $b_i \in (0,1]$ and the sum of the agents' entitlements is equal to one, $\sum_{i=1}^n b_i = 1$. Agents with higher entitlements $b_i$ stake a higher claim on the items. Also, the equal-entitlements setting considered in previous sections corresponds to $b_i = 1/n$, for all the agents.

Focusing on division of goods, we will obtain results for a shared-based fairness notion, $\mmshat$, defined in \cite{Babaioff2024ShareBasedFF}. This notion generalizes MMS to the current context of distinct entitlements. In fact, Babaioff and Feige \cite{Babaioff2024ShareBasedFF} show that $\mmshat$ \emph{dominates} other shared-based notions.   

Here, we denote a fair division instance via the tuple $\langle [n], [m], \{v_i\}_{i=1}^n, \{b_i\}_{i=1}^n \rangle$ in which $b_i$s denote the agents' entitlements and, as in Section \ref{sec:goods}, $v_i$s denote the agents' valuations. We also index the agents in decreasing order of their entitlements, $b_1 \geq b_2 \geq \ldots \geq b_n$. The entitlement-adapted share, $\widehat{\mu}_i$ of each agent $i$ is obtained by considering size $n_i = \lfloor \frac{1}{b_i} \rfloor$ partitions of the $m$ goods. Formally,  
\begin{definition}[$\mmshat$] For any fair division instance $\langle [n], [m], \{v_i\}_{i=1}^n, \{b_i\}_{i=1}^n \rangle$ with goods and entitlements, the maximin share, $\muhat_i$, of agent $i$ is defined as
$$\muhat_i \coloneqq  \max_{(X_1, \ldots X_{n_i})\in \Pi_{n_i}([m])} \ \ \min_{j \in [n_i]} v_i(X_j). $$
Here, $n_i=\lfloor\frac{1}{b_i}\rfloor$ and the maximum is taken over all $n_i$-partitions of $[m]$. Furthermore, a multi-allocation $\calA = (A_1, \ldots A_n)$ is said to be an $\mmshat$ multi-allocation if it satisfies $v_i(A_i) \geq \widehat{\mu_i}$ for each agent $i \in [n]$. \\
\end{definition}

As before, we leverage the existence of an $\mmshat$-inducing $n_i$-partition $\calM^i=(M^i_1, M^i_2, \dots, M^i_{n_i})$ for each agent $i$, given by $$\calM^i \in \argmax_{(X_1, \ldots X_{n_i})\in \Pi_{n_i}([m])} \ \  \min_{j \in [n_i]} v_i(X_j).$$

The main result of this section is as follows.

\begin{theorem}
\label{thm:mms_hat_monotone_goods}
        Every fair division instance $\langle [n], [m], \{v_i\}_{i=1}^n, \{b_i\}_{i=1}^n \rangle$ with monotone valuations and arbitrary entitlements admits an $\mmshat$ multi-allocation $\calA = (A_1,\ldots A_n)$ in which no single good is assigned to more than $3 \log m$ agents and the total number of goods assigned, with copies, is at most $\lceil 1.7  m \rceil$. That is, the characteristic vector $\chi^\calA$ of $\calA$ satisfies $\|\chi^\calA\|_\infty \leq 3 \log m$ and $\|\chi^\calA\|_1\leq \lceil 1.7 m \rceil$.
\end{theorem}

As in the proof of \Cref{thm:monotone_goods}, we proceed via the probabilistic method. The random draw here is as follows: For each agent $i\in [n]$, recall that $\calM^i=(M^i_1,\dots,M^i_{n_i})$ is an $\mmshat$-inducing partition for $i$. Independently for each $i$, we select a bundle $R_i$ uniformly at random from $\{M^i_1,\dots,M^i_{n_i}\}$, i.e., $\prob{R_i = M^i_j}=1/n_i$ for each index $j\in [n_i]$. Considering the random $\mmshat$ multi-allocation $\calR=(R_1,\dots,R_n)$, let $\chi^\calR$ be the characteristic vector of $\calR$. We first prove the following lemma which bounds the expected value of the components of the random vector $\chi^\calR$.

\begin{lemma}
    $\expval{\chi^\calR_g} \leq 1.7$ for each component (good) $g \in [m]$. \label{lem:expval-entitlements}
\end{lemma}

\begin{proof}
    For each agent $i\in [n]$ and good $g\in [m]$, let $\mathbb{1}_{\{g \in R_i\}}$ be the indicator random variable for the event $g\in R_i$. Since $R_i$ is chosen uniformly at random from the $n_i$-partition $\{M^i_1,\dots,M^i_{n_i}\}$ of $[m]$ and each good $g\in [m]$ belongs to exactly one subset $M^i_j$, we have for each fixed $i\in [n]$ and $g \in [m]$
\begin{equation}
    \expval{\mathbb{1}_{\{g\in R_i\}}} = \prob{g \in R_i} = \frac{1}{n_i}. \label{eq:indicator-exp}
\end{equation}

Further, $\chi^\calR_g$ is the number of copies of the good $g$ assigned among the agents under the multi-allocation $\calR$, i.e., $\chi^\calR_g = \sum_{i=1}^n \mathbb{1}_{\{g \in R_i\}}$.  Therefore, 
\begin{equation}
\label{eq:eq1} \expval{\chi^\calR_g} = \sum_{i=1}^n \expval{\mathbb{1}_{ \{ g \in R_i \} }} 
\underset{\text{via (\ref{eq:indicator-exp})}}{=} \sum_{i=1}^n \frac{1}{n_i}
\end{equation}

We next establish an upper bound on the right-hand-side of equation (\ref{eq:eq1}). For each agent $i \in [n]$, let $k_i$ be the unique integer such that $\frac{1}{k_i+1} < b_i \leq \frac{1}{k_i}$. For the integer $k_i$, note that the bound $k_i \leq \frac{1}{b_i} < k_i + 1$ implies  
\begin{equation}
    k_i = \left\lfloor \frac{1}{b_i} \right\rfloor = n_i \label{eq:ni-value}
\end{equation}
Further, given that $\frac{1}{b_i} < k_i +1$, we obtain  $k_i > \frac{1}{b_i} -1$ and, hence, 
\begin{equation}
    \frac{1}{n_i} = \frac{1}{k_i} < \frac{1}{\frac{1}{b_i} - 1} = \frac{b_i}{1 - b_i}. \label{eq:upperbound-bi}
\end{equation}

Recall the agents are indexed such that $b_1 \geq b_2 \geq \dots \geq b_n$. We proceed via the following case analysis.\\

\noindent\textit{Case 1}: $b_1 \leq \frac{1}{3}$. Here, for each agent $i$, the entitlement $b_i \leq \frac{1}{3}$; equivalently, $1-b_i \geq \frac{2}{3}$. Therefore, 
\begin{align*}
\sum_{i=1}^n \frac{1}{n_i} 
& < \sum_{i=1}^n \frac{b_i}{1-b_i} \tag{via (\ref{eq:upperbound-bi})} \\
&\leq \frac{3}{2} \sum_{i=1}^n b_i \tag{since $1-b_i \geq \frac{2}{3}$ for all $i$} \\
&= \frac{3}{2} \tag{since $\sum_{i=1}^n b_i = 1$}
\end{align*}

\noindent\emph{Case 2}: $\frac{1}{3} < b_1 \leq \frac{1}{2}$. By the definition of $k_i$ and equation (\ref{eq:ni-value}), we obtain $n_1= k_1 = 2$. With $b_1\geq b_2$, we consider the following two sub-cases. \\

    \noindent
    \emph{Case 2(a):} $b_2 \leq \frac{1}{3}$.  For each $i \geq 2$, we have $b_i\leq \frac{1}{3}$ and, hence, $1 - b_i \geq \frac{2}{3}$. Therefore, 
    \begin{align*}
    \sum_{i=1}^n\frac{1}{n_i} 
    & < \frac{1}{n_1} + \sum_{i=2}^n \frac{b_i}{1-b_i} \tag{via (\ref{eq:upperbound-bi})} \\ 
    & \leq \frac{1}{2} + \frac{3}{2} \sum_{i=2}^n b_i \tag{since $n_1 = 2$ and $1 - b_i> \frac{2}{3}$ for $i\geq 2$} \\
    & = \frac{1}{2} + \frac{3}{2}(1-b_1)   \tag{since $\sum_{i=1}^n b_i = 1$}\\
    & <  \frac{1}{2} + \frac{3}{2}\left(1-\frac{1}{3}\right) \tag{since $b_1>\frac{1}{3}$}\\
    &= \frac{3}{2}
    \end{align*}

    \noindent
    \emph{Case 2(b):} $\frac{1}{3} < b_2 \leq b_1 \leq \frac{1}{2}$. In this sub-case, we have $n_2 = k_2 = 2$ as well. Since $\sum_{i=1}^n b_i = 1$, for the remaining agents $i \geq 3$, we have $b_i < \frac{1}{3}$. Hence,
    \begin{align*}
    \sum_{i=1}^n\frac{1}{n_i} 
    &<  \frac{1}{n_1} + \frac{1}{n_2} + \sum_{i=3}^n \frac{b_i}{1-b_i} \tag{via (\ref{eq:upperbound-bi})} \\ 
    & < \frac{1}{2}+\frac{1}{2} + \frac{3}{2} \sum_{i=3}^n b_i \tag{since $n_1 = n_2 = 2$ and $1 - b_i> \frac{2}{3}$ for $i\geq 3$} \\
    & = 1 + \frac{3}{2}(1-b_1-b_2)  \tag{since $\sum_{i=1}^n b_i = 1$}\\
    & <  1 + \frac{3}{2} \left(1-\frac{1}{3}-\frac{1}{3}\right) \tag{since $b_1\geq b_2 > \frac{1}{3}$} \\
    &= \frac{3}{2}.
    \end{align*}

\noindent\emph{Case 3}: $b_1> \frac{1}{2}$. In this case, $n_1 = k_1 = 1$. Also, $\sum_{i=1}^n b_i=1$ implies that $b_2 < \frac{1}{2}$. Therefore, the following three sub-cases (analyzed below) are exhaustive: 

\noindent
Case 3(a): $b_2 \leq 1/4$. 

\noindent
Case 3(b): $b_2 \in (1/4, 1/3]$. 

\noindent
Case 3(c): $b_2 \in (1/3, 1/2)$.  \\

   \noindent
   \emph{Case 3(a):} $b_2 \leq \frac{1}{4}$. Here, for all $i \geq 2$, we have $1 - b_i \geq 1 - b_2 \geq 3/4$. Using this inequality, we obtain   
   \begin{align*}
    \sum_{i=1}^n\frac{1}{n_i} 
    & < \frac{1}{n_1} + \sum_{i=2}^n \frac{b_i}{1-b_i} \tag{via (\ref{eq:upperbound-bi})} \\ 
    & \leq 1 + \frac{4}{3} \sum_{i=2}^n b_i \tag{since $n_1 = 1$ and $1 - b_i> \frac{3}{4}$ for $i\geq 2$} \\
    & = 1 + \frac{4}{3}(1-b_1)   \tag{since $\sum_{i=1}^n b_i = 1$}\\
    & <  1 + \frac{4}{3}\left(1-\frac{1}{2}\right) \tag{since $b_1>\frac{1}{2}$}\\
    &= \frac{5}{3}
    \end{align*}

   \noindent
   \emph{Case 3(b):} $\frac{1}{4} < b_2 \leq \frac{1}{3}$. In this sub-case, the definition of $k_i$ gives us $n_2 = k_2 = 3$. Further, for all $i \geq 3$, we have $b_i \leq b_3 \leq 1 - b_1 - b_2 \leq 1 - \frac{1}{2} - \frac{1}{4} = \frac{1}{4}$. That is, $1 - b_i \geq \frac{3}{4}$ for each $i \geq 3$. These inequalities give s 
    \begin{align*}
    \sum_{i=1}^n\frac{1}{n_i} 
    & < \frac{1}{n_1} + \frac{1}{n_2} + \sum_{i=3}^n \frac{b_i}{1-b_i} \tag{via (\ref{eq:upperbound-bi})} \\ 
    & \leq 1 + \frac{1}{3} + \frac{4}{3} \sum_{i=3}^n b_i \tag{since $n_1 = 1$, $n_2 = 3$, and $1 - b_i> \frac{3}{4}$ for $i\geq 3$} \\
    & = 1 + \frac{1}{3} +  \frac{4}{3}(1-b_1 - b_2)   \tag{since $\sum_{i=1}^n b_i = 1$}\\
    & <  1 + \frac{1}{3} + \frac{4}{3}\left(1-\frac{1}{2} - \frac{1}{4} \right) \tag{since $b_1>\frac{1}{2}$ and $b_2 > \frac{1}{4}$}\\
    &= \frac{5}{3}
    \end{align*}

    \emph{Case 3(c):} $\frac{1}{3} < b_2 < \frac{1}{2}$. Here, we have $n_2 = k_2 = 2$. Also, the equality $\sum_{i=1}^n b_i = 1$ implies $b_3 \leq 1-b_1 -b_2 < 1-\frac{1}{2} -\frac{1}{3} =\frac{1}{6}$. Hence, $1-b_i>\frac{5}{6}$ for each $i\geq 3$. Using these bounds we obtain    
    \begin{align*}
    \sum_{i=1}^n \frac{1}{n_i} 
    &<  \frac{1}{n_1} + \frac{1}{n_2} + \sum_{i=3}^n \frac{b_i}{1-b_i} \tag{via (\ref{eq:upperbound-bi})} \\ 
    & < 1 + \frac{1}{2} + \frac{6}{5} \sum_{i=3}^n b_i \tag{since $n_1 = 1$, $n_2 = 2$, and $1 - b_i> \frac{5}{6}$ for $i\geq 3$} \\
    & = 1 + \frac{1}{2} + \frac{6}{5}(1-b_1-b_2)  \tag{since $\sum_{i=1}^n b_i = 1$}\\
    & <  1 + \frac{1}{2} + \frac{6}{5} \left(1-\frac{1}{2}-\frac{1}{3}\right) \tag{since $b_1 >\frac{1}{2}$ and $ b_2 > \frac{1}{3}$} \\
    &= 1+ \frac{1}{2} + \frac{1}{5} \nonumber \\
    &= \frac{17}{10}.  
    \end{align*}

The case analysis above shows that, for any set of entitlements, $\{b_i\}_{i=1}^n$, the sum $\sum_{i=1}^n \frac{1}{n_i}$ is at most $\max \left\{\frac{3}{2}, \frac{5}{3}, \frac{17}{10} \right\} = \frac{17}{10} = 1.7$. Hence, equation (\ref{eq:eq1}) gives us $\mathbb{E}[\chi^\calR_g] \leq 1.7$. The lemma stands proved. 
\end{proof}

Using \Cref{lem:expval-entitlements}, we now prove probability bounds for the following events: $G_1$ is the event that $\|\chi^\calR\|_\infty \leq 3 \log m$ and $G_2$ is the event that $\|\chi^\calR\|_1 \leq \lceil 1.7m \rceil$. 

\begin{lemma}
\label{lem:linftynormforb} 
The random characteristic vector $\chi^\calR$ satisfies: 
\begin{enumerate}
    \item[(i)] $\prob{G_1^c} \leq \frac{1}{m^{2}}.$
    \item[(ii)] $\prob{G_2^c} \leq \frac{1.7m}{ \lceil1.7m\rceil+1}$.
\end{enumerate}
\end{lemma}
\begin{proof} 
To prove (i), note that for a fixed good $g\in [m]$, the indicator random variables $\mathbb{1}_{\{g\in R_i\}}$ are independent across $i$s. Hence, the count $\chi^\calR_g=\sum_{i=1}^n \mathbb{1}_{\{g\in R_i\}}$ is a sum of independent Bernoulli random variables. Using Chernoff bound \cite[Theorem 4.4]{mitzenmacher2017probability}, we have $\prob{\chi^\calR_g \geq t} \leq \frac{1}{2^t}$ for any $t \geq 6 \expval{\chi^\calR_g}$. Since $\expval{\chi^\calR_g} \leq 1.7 \leq 2$ (\Cref{lem:expval-entitlements}), we can instantiate this bound with $t = \max \{12, 3 \log m \} \geq   6 \ \expval{\chi^\calR_g}$. For an appropriately large $m$, we obtain $ t = 3 \log m$ and, hence, $\prob{\chi^\calR_g \geq 3 \log m} \leq \frac{1}{m^{3}}$ for each $g\in [m]$. Write $G_{1,g}^c$ to denote the event $\{\chi^\calR_g \geq 3\log m\}$, and note that $G_1^c = \cup_{g=1}^m G_{1,g}^c$. We have  $\prob{G_{1,g}^c}\leq \frac{1}{m^{3}}$ for each $g\in [m]$. Hence, the union bound gives us $$\prob{G_1^c} = \prob{\cup_{g=1}^m G_{1,g}^c} \leq \sum_{g=1}^m \prob{G_{1,g}^c} \leq m \frac{1}{m^{3}} =\frac{1}{m^{2}}.$$ This establishes part (i) of the lemma.

For part (ii), using \Cref{lem:expval-entitlements}, we first obtain $\expval{\ellone{\chi^\calR}} 
= \expval{\sum_{g=1}^m \chi^\calR_g} 
= \sum_{g=1}^m \expval{\chi^R_g} 
\leq 1.7m$. Recall that $G_2$ is the event that $\|\chi^\calR\|_1 \leq \lceil 1.7m \rceil$. Markov's inequality and the fact that $\ellone{\chi^\calR} $ is an integer-valued random variable give us 
$$\prob{G_2^c} = \prob{\ellone{\chi^\calR} \geq \lceil 1.7m \rceil + 1} 
\leq \frac{\expval{\ellone{\chi^\calR}}}{ \lceil1.7m\rceil+1} 
\leq \frac{1.7m}{ \lceil1.7m\rceil+1}.$$ The theorem stands proved.
\end{proof}

Using the above lemma, we next prove \Cref{thm:mms_hat_monotone_goods}.

\begin{proof}[Proof of \Cref{thm:mms_hat_monotone_goods}] Recall that $G_1$ is the event that $\ellinfty{\chi^\calR} \leq 3 \log m$ and $G_2$ is the event that $\ellone{\chi^\calR} \leq \lceil 1.7m \rceil$. By \Cref{lem:linftynormforb} and the union bound, we have \begin{align*} \displaystyle \prob{G_1 \cap G_2} 
& = 1 - \prob{G_1^c \cup G_2^c} \\ 
& \geq 1 - (\prob{G_1^c} +\mathbb{P}\{G_2^c\}) \tag{Union bound} \\ 
& \geq 1-\left(\frac{1}{m^{2}} + \frac{1.7m}{\lceil1.7m\rceil+1} \right) \tag{by \Cref{lem:linftynormforb}} \\
& = \frac{\lceil 1.7m\rceil + 1 - 1.7m}{\lceil1.7m\rceil+1} - \frac{1}{m^{2}} \\
& \geq \frac{1}{\lceil1.7m\rceil+1} - \frac{1}{m^{2}} > 0.
\end{align*}

Hence, the events $G_1$ and $G_2$ hold together with positive probability for the random MMS multi-allocation $\calR = (R_1,\dots,R_n)$. This implies that there exists an MMS multi-allocation $\calA$ that satisfies the stated $\ell_\infty$ and $\ell_1$ bounds: $\|\chi^\calA\|_1\leq \lceil 1.7m \rceil $ and $\|\chi^\calA\|_\infty \leq 3 \log m$. This completes the proof. 
\end{proof}
\section{Hardness of Minimizing Assignment Multiplicity}
\label{sec:reduction}
This section establishes the computational hardness of selecting, for the agents $i \in [n]$, bundles $A_i$ from given MMS-inducing partitions $\calM^i = (M^i_1, \dots, M^i_n)$ with the objective of minimizing $\ellinfty{\chi^\calA}$, for the resulting multi-allocation $\calA=(A_1,\ldots, A_n)$. 

\begin{theorem}
\label{theorem:NPHard}
Given partitions $\calM^i = (M^i_1, \dots, M^i_n)$ for the agents $i \in [n]$, it is {\rm NP}-hard to decide whether there exists a multi-allocation $\calA = (A_1, \ldots A_n)$ with the properties that: (i) For each $i \in [n]$, the bundle $A_i = M^i_j$, for some index $j \in [n]$, and (ii) Every good is allocated to at most one agent (i.e., $A_i \cap A_j = \emptyset$ for all $i \neq j$).
\end{theorem}
\begin{proof}
   We present a reduction from the maximum independent set problem. Recall that in this {\rm NP}-hard problem we are given a graph $G = (V, E)$, along with  an integer $k \in \mathbb{Z}$, and the goal is to decide if there exists an independent set $I \subseteq V$ of size at least $k$. We first describe, the construction of the multi-allocation selection problem from any given maximum independent set instance with graph $G= (V, E)$ and threshold $k \in \mathbb{Z}_+$.    
   
We set the number of agents $n = |V|+1$ and the number of goods $m = |V|\binom{k}{2} + |E| \binom{k}{2} + |V|$. Further, for each index $i \in [k]$ and each vertex $u \in V$, we create a tuple $(i, u)$; there are $|V| k$ such tuples. Based on these tuples, we define a set family $\mathcal{F}$ that consists of size-$2$ sets either of the form $\{(i, u), (j, v)\}$, for each $(u, v) \in E$, or $\{(i, u), (j, u)\}$, for each $u \in V$. Formally, 
\begin{align*}
\mathcal{F} = \Big\{ \  \{(i, u), (j, u)\} \mid  i \neq j \text{ and } u \in V \Big\}~\bigcup~\Big\{ \ \{ (i, u), (j, v) \} \mid i \neq j \text{ and } (u, v) \in E \Big\}
\end{align*}
Note that $|\mathcal{F}| = |V|\binom{k}{2} + |E| \binom{k}{2}$.

Recall that the number of agents $n = |V| + 1$. We associate the first $|V|$ agents with the vertices of $G$ and the last agent will be indexed as $n$. We consider $m = |\mathcal{F}|+|V|$ goods. In particular, we include one good for each set $S \in \mathcal{F}$, these goods will be denoted by $\{g_S \mid S \in \mathcal{F}\}$. In addition, we introduce one good for each vertex in $G$ and denote them as $\{h_1, h_2, \ldots h_{|V|}\}$. 

Next, we define the partitions $\calM^i = (M^i_1, M^i_2, \ldots M^i_{n})$ for the first $k$ agents, i.e., for $ 1 \leq i \leq k$: 
\begin{align*}
M^i_u  & = \{g_S \mid S \in \mathcal{F} , (i, u) \in S\} \quad \text{ for $1 \leq u \leq |V|$, and } \\ 
M^i_n &  = M^i_{|V|+1}  = [m] \setminus \left( \cup_{u=1}^{|V|} \ M^i_u \right).
\end{align*}
For the remaining $n-k = |V|+1-k$ agents---i.e., for $k+1 \leq j \leq n$---we set the partitions $\calM^j = \{M^j_1, M^j_2, \ldots M^j_{n}\}$ as follows:

\begin{align*}
M^j_u  & = \{ h_u \} \quad \text{ for $1 \leq u \leq |V|$, and } \\  
M^j_n & = M^j_{|V|+1}  =  \{g_S \mid S \in \mathcal{F}\}.
\end{align*}
This completes the construction. We now prove the equivalence.

\paragraph{Forward Direction.} Suppose the given instance of maximum independent set problem is a {\rm Yes} instance. That is, there is an independent set $I \subseteq V$ of size at least $k$ in the graph $G$. Then, we consider the following multi-allocation $\calA$ in the constructed instance. For each agent $i \in [k]$, set $A_i = M^i_u$ where $u$ is the $i$th vertex in $I$. For the remaining agents $k+1 \leq j \leq n$, we set $A_j = M^j_{j-k}$.
    
We will show that under $\calA$ every good is assigned to at most one agent, i.e., the bundles $A_i$ are pairwise disjoint. First, notice that the last $n-k$ agents get distinct singleton goods from $\{h_1, h_2, \ldots h_{|V|}\}$ and all the first $k$ agents get goods only from $\{g_S \mid S \in \mathcal{F}\}$. Hence, the bundles $A_j$, for $k+1 \leq j \leq n$, are disjoint. Further, consider any pair of agents $i$ and $i'$ such that $1 \leq i \neq i' \leq k$. Suppose $A_i = M^i_u$ and $A_{i'} = M^{i'}_v$. Note that, by construction, $u, v \in I$ and $u \neq v$. Assume, towards a contradiction, that $A_i \cap A_{i'} \neq \emptyset$, i.e., $g_S \in M^i_u \cap M^{i'}_v$ for some good $g_S$. This containment implies $S = \{(i, u), (i', v)\}$. Such a set $S$ was included in $\calF$ only if $(u, v) \in E$. However, this contradicts the fact that $u$ and $v$ belong to the independent set $I$. Therefore, by way of contradiction, we get that $A_i \cap A_{i'} = \emptyset$. That is, the bundles $A_i$ are pairwise disjoint and, hence, as desired, every good is assigned to at most one agent under $\calA$. 

\paragraph{Reverse Direction.} Suppose the constructed instance admits a multi-allocation $\calA$ in which no item is assigned to more than one agent, i.e., the bundles in $\calA$ are disjoint. We will show that in such a case $G$ admits an independent set of size $k$. First, we make the following claim.

\begin{claim}
\label{claim:notlastbundle}
In the multi-allocation $\calA$, none of the first $k$ agents $i \in [k]$ receive the last subset, $M^i_n$, from their partition $\calM^i$, i.e., $A_i \neq M^i_n$ for all $i \in [k]$. 
\end{claim}
\begin{proof}
Towards a contradiction, assume that for some agent $i \in [k]$, the bundle $A_i = M^i_n$. Then, the last $(n-k) > 1$ agents $j \in \{k+1,\ldots, n\}$ can not receive  any of their first $|V|$ subsets in $\calM^j$. This follows from the fact that $M^i_n$ contains all the goods $\{h_1, \ldots, h_{|V|} \}$ and, hence, $M^i_n$ intersects with $M^j_u$ for all $j \in \{k+1,\ldots, n\}$ and $1 \leq u \leq |V|$. This leaves bundles $M^j_n$ for the agents $j \in \{k+1,\ldots, n\}$. However, these bundles intersect each other as well. Hence, by way of contradiction, we get that $A_i \neq M^i_n$ for all $i \in [k]$. 
\end{proof}

By \Cref{claim:notlastbundle}, we have that under the allocation $\calA$, each of the first $k$ agents $i \in [k]$ must receive one of their first $|V|$ subsets from their partitions $\calM^i$. Suppose $A_i = M^i_{u_i}$ for index $u_i \in V$.  
\begin{claim}\label{claim:indset}
The vertices $\{u_1, u_2, \ldots u_k\}$ form an independent set in the graph $G$.
\end{claim}
\begin{proof} First, note that $u_i \neq u_j$, for any pair $i \neq j \in [k]$. Otherwise, if $u_i = u_j = u$, then for the set $S = \{ (i, u), (j, u)\} \in \calF$, we have $g_S \in M^i_u = A_i$ and $g_S \in M^j_u = A_j$. Therefore, $g_S \in A_i \cap A_j$, which contradicts the fact that the bundles in $\calA$ are disjoint.
        
 Now, suppose $(u_i, u_j) \in E$ for some $i \neq j \in [k]$. Then, for the set $S = \{(i, u_i), (j, u_j)\} \in \calF$, consider the good $g_S$. We have $g_S \in M^i_{u_i} \cap M^j_{u_j}$, i.e., $g_S \in A_i \cap A_j$. This containment contradicts the fact that the bundles in $\calA$ are disjoint. Therefore, we have that $(u_i, u_j) \notin E$ for all pairs $i \neq j \in [k]$. Therefore, $\{u_1, u_2, \ldots u_k\}$ is an independent set in the graph $G$.  
\end{proof}
This concludes the reverse direction of the reduction and completes the proof of the theorem.
\end{proof}

\end{document}